\documentclass[english,10pt]{amsart}
\usepackage[utf8]{inputenc}
\usepackage{amsfonts}
\usepackage{amssymb, amsmath}
\usepackage{amsthm}
\usepackage{mathtools}
\usepackage{dsfont}
\usepackage{bm}
\usepackage{verbatim}
\usepackage{xcolor}
\usepackage{url}
\usepackage{subcaption}
\usepackage{float}
\usepackage{caption}
\usepackage[pdftex]{hyperref}
\hypersetup{plainpages=false,unicode=true,pdffitwindow=true}

\theoremstyle{plain}
\newtheorem{theorem}{Theorem}[section]
\newtheorem{lemma}[theorem]{Lemma}
\newtheorem{proposition}[theorem]{Proposition}
\newtheorem{corollary}[theorem]{Corollary}

\theoremstyle{definition}
\newtheorem{definition}[theorem]{Definition}

\theoremstyle{remark}
\newtheorem{remark}[theorem]{Remark}

\newcommand{\N}{\mathbb{N}}
\newcommand{\R}{\mathbb{R}}
\newcommand{\E}{\mathbb{E}}
\newcommand{\eps}{\varepsilon}

\newcommand{\abs}[1]{\left\lvert #1 \right\rvert}

\numberwithin{equation}{section}

\renewcommand{\hat}{\widehat}

\usepackage{todonotes}

\newcommand{\simon}[1]{{\leavevmode \color{red} #1}}

\usepackage[backend=biber,style=alphabetic]{biblatex}

\bibliography{literature.bib}

\title[Markovian approximations]{Markovian approximations of stochastic Volterra equations with the fractional kernel}

\author{Christian Bayer}
\address{Weierstrass Institute, Mohrenstraße 39, 10117 Berlin, Germany}
\email{christian.bayer@wias-berlin.de}

\author{Simon Breneis}
\address{Weierstrass Institute, Mohrenstraße 39, 10117 Berlin, Germany}
\email{simon.breineis@wias-berlin.de}

\date{\today}
\subjclass[2020]{91G60, 60G22, 60H35}
\keywords{Stochastic Volterra equation, Fractional kernel, Rough volatility model, Markovian approximation, Strong error}

\thanks{C.B. and S.B. gratefully acknowledge the support of the DFG through the IRTG 2544. C.B. additionally acknowledges the support of the DFG through the FOR 2402. Both authors would also like to thank Ahmed Kebaier for the insightful discussions.}

\begin{document}
\maketitle

\begin{abstract}
We consider rough stochastic volatility models where the variance process satisfies a stochastic Volterra equation with the fractional kernel, as in the rough Bergomi and the rough Heston model. In particular, the variance process is therefore not a Markov process or semimartingale, and has quite low Hölder-regularity. In practice, simulating such rough processes thus often results in high computational cost. To remedy this, we study approximations of stochastic Volterra equations using an $N$-dimensional diffusion process defined as solution to a system of ordinary stochastic differential equation. If the coefficients of the stochastic Volterra equation are Lipschitz continuous, we show that these approximations converge strongly with superpolynomial rate in $N$. Finally, we apply this approximation to compute the implied volatility smile of a European call option under the rough Bergomi and the rough Heston model.
\end{abstract}

\section{Introduction}

Suppose we are given a stochastic Volterra equation of the form 
\begin{equation}\label{eqn:SVE}
X_t = x_0 + \int_0^t G(t-s) b(X_s) ds + \int_0^t G(t-s) \sigma(X_s) dW_s,
\end{equation}
where $x_0 \in \R^d$, $b:\R^d\to\R^d$ and $\sigma:\R^{d\times d}\to\R^{d\times d}$ are Lipschitz continuous, $W$ is a $d$-dimensional Brownian motion. $G$ denotes a -- for simplicity -- one-dimensional kernel assumed to be completely monotone. The kernel $G$ is often a power law kernel, $G(t) \simeq t^{H-1/2}$. In analogy to fractional Brownian motion (fBm), we call the parameter $0 < H < 1$ \emph{Hurst parameter}. Indeed, the special case $\widetilde{W}_t \coloneqq \sqrt{2H} \int_0^t (t-s)^{H-1/2} dW_s$ defines the so-called \emph{Riemann--Liouville fBm}.

Models of this type are used in different applications -- we refer to \cite{marie2013pathwise} and \cite{chevillard2017regularized} for applications in biology and turbulence, respectively. We are mostly motivated by the recent framework of \emph{rough stochastic volatility models} in finance. In this context, we typically consider an asset price process
\begin{equation*}
  dS_t = \sqrt{V_t} S_t d Z_t
\end{equation*}
defined in terms a a Brownian motion (Bm) $Z$ and the \emph{stochastic variance} process $v$ -- we consider dynamics under a risk neutral measure and assume interest rate $r = 0$ for simplicity. As observed in \cite{gatheral2018volatility} (for time series of stock prices) and \cite{B2016} (for option prices), stochastic volatility models can provide excellent fits to market data when the stochastic variance follows a dynamic of stochastic Volterra type with fractional kernel an \emph{small $H \approx 0$} -- suggesting the name ``rough stochastic volatility''. A popular example of a rough stochastic volatility model is the \emph{rough Heston model}, see \cite{E2018}, given by the dynamics
\begin{equation*}
  V_t = V_0 + \int_0^t G(t-s) \big(\theta -\lambda V_s\big) ds + \int_0^t G(t-s) \nu\sqrt{V_s} dB_s.
\end{equation*}
Another important example is the \emph{rough Bergomi model}, see \cite{B2016},
\begin{equation*}
  V_t \coloneqq \xi(t) \exp\bigg(\eta\sqrt{2H}\int_0^t (t-s)^{H-1/2}dB_s - \frac{\eta^2}{2}t^{2H}\bigg),
\end{equation*}
where $\xi(t)$ is the forward variance curve. In both cases, the Bm $B$ driving the variance dynamics is correlated with the Bm $Z$ driving $S$.

While very useful from a modeling point of view, the fact that rough stochastic volatility processes -- or, more generally, stochastic Volterra processes -- are no Markov processes leads to serious theoretical and numerical challenges. This motivates the use of \emph{Markovian approximations} to the stochastic Volterra processes. A simple and flexible approach goes back to \cite{carmona1998fractional}. It is well-known that any \emph{completely monotone} kernel $G$ allows for a representation in terms of a Laplace transform, i.e., there is a measure $\mu$ on the positive half-line such that
\begin{equation*}
  G(t) = \int_0^\infty e^{-tx} \mu(dx), \quad t > 0.
\end{equation*}
Hence, the stochastic Fubini theorem implies that
\begin{equation*}
  \widetilde{W}_t \coloneqq \int_0^t G(t-s) dW_s = \int_0^\infty \int_0^t e^{-(t-s) x} dW_s \, \mu(dx) = \int_0^t Y_t^x \mu(dx), \quad Y^x_t \coloneqq \int_0^t e^{-(t-s) x} dW_s,
\end{equation*}
where $Y^x$ is an Ornstein--Uhlenbeck process, and, hence, a Markov process. Indeed, the infinite-dimensional process $\mathcal{Y}_t \coloneqq (Y^x_t)_{x>0}$ is a Markov process with state space $L^1(\R_{>0}, \mu)$, see \cite{carmona1998fractional}.

In what follows, we will mostly restrict ourselves to the fractional kernel given by
\begin{equation}
  \label{eq:frac-kernel}
  G(t) \coloneqq \frac{t^{H-1/2}}{\Gamma(H+1/2)},
\end{equation}
where we will assume that $H\in(0, 1/2)$. In this case, we have 
\begin{equation}\label{eqn:LaplaceTransformKernel}
G(t) = c_H \int_0^\infty e^{-xt} x^{-H-1/2} dx,\qquad c_H \coloneqq \frac{1}{\Gamma(H+1/2)\Gamma(1/2-H)},
\end{equation}
i.e., $\mu(dx) = w(x) dx$ with $w(x) \coloneqq c_H x^{-H-1/2}$, $x > 0$.
Approximating this integral by a finite sum, i.e., $G \approx \hat{G}$ with 
\begin{equation}\label{eqn:ApproximatedKernel}
\hat{G}(t) \coloneqq \sum_{i=1}^N w_ie^{-x_i t},
\end{equation}
we get
\begin{equation}\label{eqn:ApproximatedSVE}
\hat{X}_t = x_0 + \int_0^t \hat{G}(t-s) b(\hat{X}_s) ds + \int_0^t \hat{G}(t-s) \sigma(\hat{X}_s) dW_s.
\end{equation}

Such approximations were already considered by, e.g., Abi Jaber and El Euch \cite{A2019}, Alfonsi and Kebaier \cite{A2021}, and, for the particular case of fractional Brownian motion, by Harms \cite{H2019}. Indeed, we have the following proposition, which shows that~\eqref{eqn:ApproximatedSVE} reduces to an $N$-dimensional ordinary stochastic differential equation, instead of a stochastic Volterra equation.

\begin{proposition}\cite[Proposition 2.1]{A2021}.
Let $x_0^1, \dots, x_0^N\in\R^d$ be such that $\sum_{i=1}^N w_i x^i_0 = x_0.$ Then the solution to \eqref{eqn:ApproximatedSVE} is given by $\sum_{i=1}^N w_i X^i_t,$ where $(X^1_t,\dots,X^N_t)$ is the solution to the $(N\times d)$-dimensional SDE defined by
\begin{equation}\label{eqn:SDEApproximationOfSVE}
X^i_t = x^i_0 - \int_0^t x_i (X^i_t - x^i_0) ds + \int_0^t b\bigg(\sum_{j=1}^N w_j X^j_s\bigg) ds + \int_0^t \sigma\bigg(\sum_{j=1}^N w_j X^j_s\bigg) dW_s.
\end{equation}
\end{proposition}

Perhaps unsurprisingly, the approximation error between $X$ and the Markovian approximation $\hat{X}$ is controlled by the error between $G$ and $\hat{G}$, more precisely by the $L^2$-error. The following bound is a slight modification of~\cite[Proposition 3.2]{A2021}.

\begin{proposition}\label{prop:AKCorollary}\cite[Proposition 3.2]{A2021}.
For every $T>0$, there exists a constant $C$ (depending on $T$, $|x_0|$, $b$, $\sigma$), such that $$\E\big|\hat{X}_T - X_T \big|^2 \le C\int_0^T \big|G(t) - \hat{G}(t)\big|^2 dt.$$
\end{proposition}

Using this bound, Alfonsi and Kebaier in \cite[Corollary 3.1]{A2021} constructed a point set $(x_i)_{i=1,\dots,N}$ with weights $(w_i)_{i=1,\dots,N}$ by cutting off the integral in \eqref{eqn:LaplaceTransformKernel} at some large value $K$, and then using the midpoint rule on linearly spaced intervals on $[0, K]$. They were able to prove a strong rate of convergence of $N^{-2H/3}.$ In the same paper, they were able to improve the rate of convergence to almost $N^{-H}$ using some more sophisticated constructions. This is a very slow convergence rate, in particular as $H$ is often found to be close to $0$ in empirical studies, see \cite{B2016}. For instance, if $H = 0.1$, increasing the accuracy of the computation by one more significant digit would require increasing $N$ by a factor $10^{10}$. The aim of this paper is to considerably improve the rate of convergence by using point sets adapted to the problem at hand.

In this we are motivated by Harms \cite{H2019}. Rather than linearly spaced intervals, Harms used a geometric spacing. Additionally, he used Gaussian quadrature rules of arbitrary level $m$ instead of the midpoint rule for the approximation of the integral in \eqref{eqn:LaplaceTransformKernel}. Using this approach, Harms managed to prove a rate of convergence of almost $n^{-2Hm/3}$, where $n=N/m$ is the number of intervals used (see \cite[Theorem 1]{H2019}). While this yields a convergence of arbitrary order, it is still apparent that $m$ has to be chosen quite large to get a suitably large order if $H$ is small. This can of course pose additional problems, as the constants in his bound likely grow in $m$. It should also be noted that Harms did not give a suggestion on which $m$ to use given a choice of $N$, and also, that he chose a quite different approach and only proved this rate of convergence for the particular case of fractional Brownian motion. We also remark that, to the best of our knowledge, Harms was the first who considered a geometric spacing of the intervals for the approximation of the kernel.

We will combine the ideas of these papers by using the estimate of Proposition \ref{prop:AKCorollary} by Alfonsi and Kebaier, and a modified version of the point set used by Harms. We obtain an error bound of the form $\exp\left( - c \sqrt{N} \right)$. For more details on the rate of convergence together with an explicit choice of the points $x_i$ and weights $w_i$ see Theorem~\ref{thm:TheGrandAKTheorem}. A further improvement (in terms of constants) is given in Theorem \ref{thm:TheGrandAKTheoremWithConstants}. These results are then applied to the rough Bergomi model in Section \ref{sec:rBergomi} -- in the context of simulation methods -- and the rough Heston model in Section~\ref{sec:rHeston} -- regarding the characteristic function and Fourier pricing methods.

Numerical experiments using the various suggested point sets for the simulation of fractional Brownian motion are conducted in Section~\ref{sec:numerics}. Here, these approximations are also compared to the point sets suggested by \cite{A2021} and \cite{H2019}. Further improvements in the choices of points $x_i$ and weights $w_i$ are discussed in Section \ref{sec:LearningBetterRate} based on empirical optimizations of constants. Finally, the point set suggested in Section \ref{sec:LearningBetterRate} is applied for the approximation of the implied volatility smile under the rough Bergomi model in Section \ref{sec:rBergomiSmile} and for the implied volatility smile under the rough Heston model in Section \ref{sec:rHestonSmile}. For both models, we also compare our approximation to already existing ones.

It should be noted that similar results can be obtained for other completely monotone kernels $G$ corresponding to other weight functions $w$. Indeed, the rates of convergence observed in this paper are determined by the rate of explosion of $w$ at $0$ as well as the rate of decay at $\infty$, as well as its regularity on the full domain. An abstract result of this form is provided in Theorem \ref{thm:TheGrandAKTheoremGeneral} below. However, for the delicate task of finding \emph{good} point sets and weights, we prefer to work with the specific class of fractional kernel.

\section{Main results}

In this paper, we will write $a \approx b$ for $a\in\N$ and $b\in\R$ if $\abs{a - b} < 1$, i.e., if $a$ can be obtained from $b$ by rounding to one of the two nearest integers. This notation will be used in various error estimation or optimization procedures, when integer-valued variables will be relaxed to real ones for simplification.

\subsection{Superpolynomial rate of convergence}\label{sec:SuperpolynomialRateOfConvergence}

We now describe the nodes and weights used in our approximation of the kernel $G$. They essentially follow a Gaussian quadrature rule. A reminder of Gaussian quadrature is given in Appendix \ref{sec:GaussianQuadrature}.

Let $N\in\N$ the total number of nodes and $\alpha,\beta,a,b\in(0,\infty)$ be parameters of the scheme. Together with $A = A_H \coloneqq \left(\frac{1}{H} + \frac{1}{3/2-H}\right)^{1/2}$, where $0 < H <1/2$ denotes the Hurst index of the kernel $G$, we define
\begin{gather*}
  m :\approx \frac{\beta}{A}\sqrt{N}, \quad n :\approx \frac{A}{\beta}\sqrt{N}\ (\approx N/m),\\
  \xi_0 \coloneqq a\exp\bigg(-\frac{\alpha}{(3/2-H)A}\sqrt{N}\bigg), \quad \xi_n \coloneqq b\exp\bigg(\frac{\alpha}{HA}\sqrt{N}\bigg),\\
\xi_i \coloneqq \xi_0\bigg(\frac{\xi_n}{\xi_0}\bigg)^{i/n},\quad i=0,\dots,n.
\end{gather*}

Then, we define the Gaussian rule of type $(H,N,\alpha,\beta,a,b)$ to be the set of nodes $(x_i)_{i=1}^{nm}$ with weights $(w_i)_{i=1}^{nm}$ of the Gaussian quadrature rule of level $m$ applied to the intervals $[\xi_i,\xi_{i+1}]$ for $i=0,\dots,n-1$ with the weight function
\begin{equation}\label{eqn:WeightFunction}
w(x) \coloneqq c_H x^{-H-1/2},\quad c_H \coloneqq \frac{1}{\Gamma(H+1/2)\Gamma(1/2-H)}.
\end{equation}
In addition, we set $x_0 \coloneqq 0$ and $$w_0 \coloneqq c_H \int_0^{\xi_0} x^{-H-1/2} dx = \frac{c_H}{1/2-H}\xi_0^{1/2-H}.$$ Whether in practice we round up or down in the definition of $m$ depends on the context and will always be stated explicity. The value of $n$ is then chosen such that $nm$ is as close as possible to $N$. The approximation of $G$ is then given by
\begin{equation}
  \hat{G}(t) \coloneqq \sum_{i=0}^{nm} w_i e^{-x_it}.\label{eq:def-hatG}
\end{equation}

\begin{theorem}\label{thm:TheGrandAKTheorem}
Let $x_0\in\R^d$, and let $b:\R^d\to\R$ and $\sigma:\R^{d\times d}\to\R^{d\times d}$ be globally Lipschitz continuous functions. Let $X$ be the solution of \eqref{eqn:SVE}, and let $\hat{X}$ be the solution of \eqref{eqn:ApproximatedSVE}, where we use the Gaussian rule of type $(H,N,\alpha,\beta,1,1)$ with $\alpha \coloneqq 1.06418$ and $\beta \coloneqq 0.4275$. Then, $$\E\big|X_T - \hat{X}_T\big|^2 \le Cc_H^2 \bigg(\frac{T^3}{(3/2-H)^2} + \frac{3}{2H^2} + \frac{5\pi^3}{48}\Big(e^{\alpha\beta}-1\Big)^2\frac{A^{2-2H}T^{2H}}{\beta^{2-2H}H}N^{1-H}\bigg) \exp\bigg(-\frac{2\alpha}{A}\sqrt{N}\bigg),$$ where $C$ is the constant from Proposition \ref{prop:AKCorollary}.
\end{theorem}
Note that we have $\alpha/A \sim \alpha\sqrt{H}$ as $H\to 0$. For $H=0.1$, $\alpha/A \approx 0.3251$.

\begin{remark}
Theorem \ref{thm:TheGrandAKTheorem} is not only true for the specific values of $\alpha$ and $\beta$ given above, but for an entire set $\Gamma\subset\R_+^2$ of pairs $(\alpha,\beta)$. The theorem thus tells us that in particular, $(1.06418, 0.4275)\in\Gamma$. While this parameter choice is the best for which we were able to show the theorem, we will observe in Section \ref{sec:LearningBetterRate} that in practice we may achieve even better results with a different choice of $(\alpha,\beta)$.
\end{remark}

\begin{remark}
While we apply Gaussian quadrature rules of level $m$ on the intervals $[\xi_i, \xi_{i+1}]$, we only use a Riemann-type approximation on $[0, \xi_0]$. This has a couple of reasons. First, it simplifies the proof. Second, we are interested in rough volatility models, where one usually needs to jointly simulate the Volterra process and the underlying Brownian motion (for the stock price). As taking $x_0=0$ corresponds to no mean reversion, adding this node often does not increase the computational cost at all. Third, while using a Gaussian quadrature rule on $[0,\xi_0]$ could potentially improve the error up to $$\exp\bigg(-\alpha\sqrt{H}\sqrt{N}\bigg),$$ this is really only a minor asymptotic speedup. Indeed, for $H=0.1$, we have $\alpha\sqrt{H} \approx 0.3365$, which is only insignificantly larger than $\alpha/A \approx 0.3251.$ At the same time, this would add $m$ additional nodes that are (compared to the node $x_0 = 0$) non-trivial, increasing the computational cost especially for small $N$. Numerical experiments indicate that using a Gaussian quadrature rule on $[0,\xi_0]$ only becomes worthwile for sufficiently large $N$, say $N>100$.
\end{remark}

\begin{remark}
It is well-known that the weights obtained by Gaussian quadrature are non-negative. This guarantees that the approximated kernel $\widehat{G}$ is completely monotone, just like the original kernel $G$. In particular, this ensures the non-negativity of the approximating volatility process for the rough Heston model, see \cite[Theorem 3.1]{A2019}.
\end{remark}

\begin{remark}\label{rmk:RoundmUp}
The bound in Theorem \ref{thm:TheGrandAKTheorem} is completely non-asymptotic. For simplicity, we assume that $m$ and $n$ are real-valued, rather than integer-valued. One can still take the suggested point set of the theorem by simply rounding $m$ and $n$. As the theorem seems to overestimate the quadrature error (see Section \ref{sec:OptimalChoices}), resulting in artificially small intervals $[\xi_0, \xi_n]$, we recommend always rounding $m$ up to the next integer. (Note that we do not make the same recommendation for the point set suggested in Section \ref{sec:LearningBetterRate}.)
\end{remark}

\begin{remark}
We are not sure whether the geometric choice of the quadrature intervals $[\xi_i,\xi_{i+1}]$ is optimal. While with this choice of intervals some expressions simplify nicely in the proof of the upcoming Lemma \ref{lem:IntegrationErrorOnAllIntervals}, it might be possible to achieve a better rate of convergence with a smarter choice of intervals.
\end{remark}

\begin{remark}\label{rmk:Thisw0NotOptimal}
In practice, we actually propose using a weight $w_0$ different from the one stated in Theorem \ref{thm:TheGrandAKTheorem}, see Remark \ref{rmk:BetterW0}. The weight $w_0$ suggested above is only used to simplify the proof of the theorem.
\end{remark}

Before proceeding to the proof of Theorem \ref{thm:TheGrandAKTheorem}, we prove some auxiliary lemmas on the errors of Gaussian quadrature. In these proofs, we will make use of results in Appendix \ref{sec:GaussianQuadrature}.

\begin{lemma}\label{lem:IntegrationErrorOnOneInterval}
Let $\tilde{w}_i$ be the weights and $\tilde{x}_i$ be the nodes of the Gaussian quadrature rule for $i=1, \dots, m$ on the interval $[a,b]$ with respect to the weight function $w(x) = c_H x^{-H-1/2}.$ Then, $$\Bigg|c_H\int_a^b e^{-tx} x^{-H-1/2} dx - \sum_{i=1}^m \tilde{w}_i e^{-t\tilde{x}_i}\Bigg| \le \sqrt{\frac{5\pi^3}{18}}\frac{c_H}{2^{2m+1}m^H} t^{-1/2+H}\bigg(\frac{b}{a}-1\bigg)^{2m+1}.$$
\end{lemma}

\begin{proof}
By Lemma \ref{lem:PeanoRepresentation} and Lemma \ref{lem:PeanoKernelBound},
\begin{align*}
\Bigg|c_H\int_a^b e^{-tx} x^{-H-1/2}& dx - \sum_{i=1}^m \tilde{w}_i e^{-t\tilde{x}_i}\Bigg| = \Bigg|\int_a^b \partial_x^{2m}e^{-tx} K_{2m}(x) dx\Bigg| \\
&\le t^{2m} \int_a^b e^{-tx}|K_{2m}(x)| dx\\
&\le t^{2m} \int_a^b e^{-ta}\frac{(2\pi)^{2m}}{(2m)!}\bigg(\frac{b-a}{2}\bigg)^{2m} \sup_{y\in[-1,1]} |B_{2m}(y)| \sup_{y\in[a,b]} |w(y)| dx.
\end{align*}

Note that for even $s$, $B_s$ is an even function, implying that $\sup_{y\in[-1,1]} |B_s(y)| = \sup_{y\in[0,1]} |B_s(y)|.$ Hence, we can apply Lemma \ref{lem:BernoulliBound} to get
\begin{multline*}
\Bigg|c_H\int_a^b e^{-tx} x^{-H-1/2} dx - \sum_{i=1}^m \tilde{w}_i e^{-t\tilde{x}_i}\Bigg| \le\\ t^{2m} \int_a^b e^{-ta}\frac{(2\pi)^{2m}}{(2m)!}\bigg(\frac{b-a}{2}\bigg)^{2m} \frac{2\zeta(2m)}{(2\pi)^{2m}}c_H a^{-H-1/2} dx.
\end{multline*}
Using that $\zeta(2m) \le \zeta(2) = \pi^2/6$ yields
\begin{multline*}
\Bigg|c_H\int_a^b e^{-tx} x^{-H-1/2} dx - \sum_{i=1}^m \tilde{w}_i e^{-t\tilde{x}_i}\Bigg| \le t^{2m} e^{-ta}\frac{1}{(2m)!}\frac{(b-a)^{2m+1}}{2^{2m}} \frac{\pi^2}{3}c_H a^{-H-1/2}\\
= \frac{\pi^2 c_H}{3\cdot 2^{2m}(2m)!}t^{2m} e^{-ta}a^{2m+1/2-H}\bigg(\frac{b}{a}-1\bigg)^{2m+1}.
\end{multline*}

Stirling's formula $k! \ge \sqrt{2\pi k} \left(\frac{k}{e}\right)^k$ and the estimate $e^{-x} \le \left(\frac{\eta}{e}\right)^\eta x^{-\eta}$, $x,\eta\ge 0$ for $\eta = 2m+1/2-H$ yield
\begin{align*}
\Bigg|c_H\int_a^b &e^{-tx} x^{-H-1/2} dx - \sum_{i=1}^m \tilde{w}_i e^{-t\tilde{x}_i}\Bigg| \\
&\le \frac{\pi^2 c_H}{3\cdot 2^{2m}2\sqrt{\pi m}}\bigg(\frac{e}{2m}\bigg)^{2m} t^{2m} \bigg(\frac{2m+1/2-H}{e}\bigg)^{2m+1/2-H} \times\\ &\qquad \times (ta)^{-2m-1/2+H}a^{2m+1/2-H}\bigg(\frac{b}{a}-1\bigg)^{2m+1}\\
&= \frac{\pi^2 c_H}{3\cdot 2^{2m+1}\sqrt{\pi m}}\bigg(\frac{2m+1/2-H}{2m}\bigg)^{2m} \bigg(\frac{2m+1/2-H}{e}\bigg)^{1/2-H} t^{-1/2+H}\bigg(\frac{b}{a}-1\bigg)^{2m+1}\\
&\le \frac{\pi^2 c_H}{3\cdot 2^{2m+1}\sqrt{\pi m}}e^{1/2-H} \bigg(\frac{2m+1/2-H}{e}\bigg)^{1/2-H} t^{-1/2+H}\bigg(\frac{b}{a}-1\bigg)^{2m+1}\\
&\le \frac{\pi^2 c_H}{3\cdot 2^{2m+1}m^H} \bigg(\frac{5/2}{\pi}\bigg)^{1/2} t^{-1/2+H}\bigg(\frac{b}{a}-1\bigg)^{2m+1}\\
&\le \sqrt{\frac{5\pi^3}{18}}\frac{c_H}{2^{2m+1}m^H} t^{-1/2+H}\bigg(\frac{b}{a}-1\bigg)^{2m+1}. \qedhere
\end{align*}
\end{proof}

The previous lemma is an error estimate of our integration error on an interval $[\xi_i,\xi_{i+1}]$. We will now combine the errors over all such intervals.

\begin{lemma}\label{lem:IntegrationErrorOnAllIntervals}
In the setting of Theorem \ref{thm:TheGrandAKTheorem}, we have $$\int_0^T \bigg|c_H \int_{\xi_0}^{\xi_n} e^{-tx} x^{-H-1/2} dx - \sum_{i=1}^N w_i e^{-tx_i}\bigg|^2 dt \le \frac{5\pi^3}{36}\frac{c_H^2T^{2H}}{H}\frac{n^2}{m^{2H}} \bigg(\frac{1}{2}\Big(e^{\alpha\beta}-1\Big)\bigg)^{4m+2}.$$
\end{lemma}

\begin{proof}
Since we split up the interval $[\xi_0,\xi_n]$ into $n$ subintervals and apply a Gaussian quadrature rule on each of them, we have with the triangle inequality and Lemma \ref{lem:IntegrationErrorOnOneInterval}
\begin{multline*}
\int_0^T \bigg|c_H \int_{\xi_0}^{\xi_n} e^{-tx} x^{-H-1/2} dx - \sum_{i=1}^N w_i e^{-tx_i}\bigg|^2 dt\\
\le \int_0^T \bigg(\sum_{i=0}^{n-1} \sqrt{\frac{5\pi^3}{18}}\frac{c_H}{2^{2m+1}m^H} t^{-1/2+H}\bigg(\frac{\xi_{i+1}}{\xi_i}-1\bigg)^{2m+1}\bigg)^2 dt.
\end{multline*}
Recall that
\begin{equation*}
\frac{\xi_{i+1}}{\xi_i} = \bigg(\frac{\xi_n}{\xi_0}\bigg)^{1/n} = \exp\Bigg(\frac{\alpha\sqrt{N}}{An}\bigg(\frac{1}{H} + \frac{1}{3/2-H}\bigg)\Bigg) = \exp\bigg(\frac{\alpha A\sqrt{N}}{n}\bigg)
= \exp\bigg(\frac{\alpha\beta A\sqrt{N}}{A\sqrt{N}}\bigg) = e^{\alpha\beta}.
\end{equation*}
Thus,
\begin{align*}
\int_0^T \bigg|c_H \int_{\xi_0}^{\xi_n} &e^{-tx} x^{-H-1/2} dx - \sum_{i=1}^N w_i e^{-tx_i}\bigg|^2 dt\\
&\le \int_0^T \bigg(\sum_{i=0}^{n-1} \sqrt{\frac{5\pi^3}{18}}\frac{c_H}{2^{2m+1}m^H} t^{-1/2+H}\Big(e^{\alpha\beta}-1\Big)^{2m+1}\bigg)^2 dt\\
&= \int_0^T \bigg(n \sqrt{\frac{5\pi^3}{18}}\frac{c_H}{2^{2m+1}m^H} t^{-1/2+H}\Big(e^{\alpha\beta}-1\Big)^{2m+1}\bigg)^2 dt\\
&= \frac{5\pi^3}{36}\frac{c_H^2T^{2H}}{H}\frac{n^2}{m^{2H}} \bigg(\frac{1}{2}\Big(e^{\alpha\beta}-1\Big)\bigg)^{4m+2}.\qedhere
\end{align*}
\end{proof}

We now move to the proof of Theorem \ref{thm:TheGrandAKTheorem}.

\begin{proof}[Proof of Theorem \ref{thm:TheGrandAKTheorem}]
  By Proposition \ref{prop:AKCorollary}, we have $$\E\big|X_T - \hat{X}_T\big|^2 \le C \int_0^T \big|G(t) - \hat{G}(t)\big|^2 dt.$$ 
  Hence,
\begin{align*}
\E\big|X_T - \hat{X}_T\big|^2 &\le C \int_0^T \bigg|c_H\int_0^\infty e^{-tx} x^{-H-1/2} dx - \sum_{i=0}^N w_i e^{-tx_i}\bigg|^2 dt\\
&\le 3C \int_0^T \Bigg(\bigg|c_H\int_0^{\xi_0} e^{-tx} x^{-H-1/2} dx - w_0\bigg|^2\\
&\qquad + \bigg|c_H\int_{\xi_0}^{\xi_n} e^{-tx}x^{-H-1/2} dx - \sum_{i=1}^N w_i e^{-tx_i}\bigg|^2\\
&\qquad + \bigg|c_H\int_{\xi_n}^\infty e^{-tx}x^{-H-1/2}dx\bigg|^2\Bigg) dt.
\end{align*}

Let us consider the first summand first. By the choice of $w_0$, we have $$\bigg|c_H\int_0^{\xi_0} e^{-tx} x^{-H-1/2} dx - w_0\bigg| = c_H\int_0^{\xi_0} \Big(1-e^{-tx}\Big) x^{-H-1/2} dx.$$ Since $e^y \ge 1 + y,$ we have
\begin{align*}
c_H\int_0^{\xi_0} \Big(1-e^{-tx}\Big) x^{-H-1/2} dx &\le c_H\int_0^{\xi_0} tx x^{-H-1/2} dx\\
&= \frac{c_H}{3/2-H} t \xi_0^{3/2-H}\\
&= \frac{c_H}{3/2-H} t \exp\bigg(-\frac{\alpha}{A}\sqrt{N}\bigg).
\end{align*}
Squaring and integrating gives $$\int_0^T \bigg|c_H\int_0^{\xi_0} e^{-tx} x^{-H-1/2} dx - w_0\bigg|^2dt \le \frac{c_H^2T^3}{3(3/2-H)^2}  \exp\bigg(-\frac{2\alpha}{A}\sqrt{N}\bigg).$$

Now, consider the last summand. We have 
\begin{align*}
\int_0^T\bigg(c_H\int_{\xi_n}^\infty e^{-tx}x^{-H-1/2}dx\bigg)^2 dt &= c_H^2\int_0^T\int_{\xi_n}^\infty\int_{\xi_n}^\infty e^{-t(x+y)}x^{-H-1/2}y^{-H-1/2}dydxdt\\
&\le c_H^2\int_{\xi_n}^\infty\int_{\xi_n}^\infty \int_0^\infty e^{-t(x+y)}dtx^{-H-1/2}y^{-H-1/2}dydx\\
&= c_H^2\int_{\xi_n}^\infty\int_{\xi_n}^\infty \frac{x^{-H-1/2}y^{-H-1/2}}{x+y}dydx\\
&\le \frac{c_H^2}{2}\int_{\xi_n}^\infty\int_{\xi_n}^\infty \frac{x^{-H-1/2}y^{-H-1/2}}{\sqrt{xy}}dydx\\
&= \frac{c_H^2}{2H^2}\xi_n^{-2H}\\
&= \frac{c_H^2}{2H^2}\exp\bigg(-\frac{2\alpha}{A}\sqrt{N}\bigg).
\end{align*}

Putting these estimates and Lemma \ref{lem:IntegrationErrorOnAllIntervals} together, we get
\begin{align}
\E\big|X_T - \hat{X}_T\big|^2 &\le 3C \Bigg(\frac{c_H^2T^3}{3(3/2-H)^2}  \exp\bigg(-\frac{2\alpha}{A}\sqrt{N}\bigg) \nonumber\\
&\qquad + \frac{5\pi^3}{36}\frac{c_H^2T^{2H}}{H}\frac{n^2}{m^{2H}} \bigg(\frac{1}{2}\Big(e^{\alpha\beta}-1\Big)\bigg)^{4m+2} \nonumber\\
&\qquad + \frac{c_H^2}{2H^2}\exp\bigg(-\frac{2\alpha}{A}\sqrt{N}\bigg)\Bigg) \nonumber\\
&= Cc_H^2 \Bigg(\bigg(\frac{T^3}{(3/2-H)^2} + \frac{3}{2H^2}\bigg) \exp\bigg(-\frac{2\alpha}{A}\sqrt{N}\bigg) \label{eqn:SomeWeirdBound}\\
&\qquad + \frac{5\pi^3}{12}\frac{A^{2-2H}T^{2H}}{\beta^{2-2H}H}N^{1-H} \bigg(\frac{1}{2}\Big(e^{\alpha\beta}-1\Big)\bigg)^{4m+2}\Bigg). \nonumber
\end{align}

Notice that
\begin{equation}\label{eqn:RewritingWeirdExpression}
\bigg(\frac{1}{2}\Big(e^{\alpha\beta}-1\Big)\bigg)^{4m} = \exp\Bigg(\log\bigg(\frac{1}{2}\Big(e^{\alpha\beta}-1\Big)\bigg)\frac{4\beta}{A}\sqrt{N}\Bigg).
\end{equation}

Our next goal is to choose $\alpha$ and $\beta$ in such a way as to maximize the rate of convergence in $N$. The rate of convergence in the first term of \eqref{eqn:SomeWeirdBound} is of course larger the larger $\alpha$ is, indicating that we would like to choose $\alpha$ as large as possible. However, larger $\alpha$ at the same time leads to slower rate of convergence in the second term of \eqref{eqn:SomeWeirdBound}, which is equivalent to the right-hand side of \eqref{eqn:RewritingWeirdExpression}. Hence, to maximize the overall rate, we first set the rates of the first term in \eqref{eqn:SomeWeirdBound} equal to the rate of the second term in \eqref{eqn:SomeWeirdBound}, and afterwards, we maximize over $\alpha$. In all that, we ignore the factor $N^{1-H}$ in the second term of \eqref{eqn:SomeWeirdBound}, as it is comparatively negligible. We therefore have to solve the optimization problem
\begin{equation}
  \label{eqn:OptimizationProblem}
\alpha \to \max!, \text{ subject to }
-\frac{2\alpha}{A}\sqrt{N} = \log\bigg(\frac{1}{2}\Big(e^{\alpha\beta}-1\Big)\bigg)\frac{4\beta}{A}\sqrt{N},
\end{equation}
where we maximize over $\alpha, \beta>0$.

As one can see, this optimization problem is completely ``dimensionless'', so the solutions $\alpha,\beta$ are actually numbers independent of any parameters, and are given as in the statement of the theorem. For these values, we have
\begin{multline*}
  \E\big|X_T - \hat{X}_T\big|^2 \le \\
  Cc_H^2 \bigg(\frac{T^3}{(3/2-H)^2} + \frac{3}{2H^2} + \frac{5\pi^3}{48}\Big(e^{\alpha\beta}-1\Big)^2\frac{A^{2-2H}T^{2H}}{\beta^{2-2H}H}N^{1-H}\bigg) \exp\bigg(-\frac{2\alpha}{A}\sqrt{N}\bigg).
\end{multline*}
This proves the theorem.
\end{proof}

\subsection{Markovian approximation for general completely monotone kernels}
\label{sec:mark-appr-gener}

We remark that a proof similar to the one of Theorem \ref{thm:TheGrandAKTheorem} can be done for more general completely monotone kernels. In Theorem \ref{thm:TheGrandAKTheoremGeneral} we give one such generalization. A rough sketch of the proof and the precise choice of nodes is given in Appendix \ref{sec:ProofOfTheoremGeneral}. We consider the following class of kernels.

\begin{definition}
A completely monotone kernel $G$ is said to be of type $(\gamma,\delta)$ if $$G(t) = \int_0^\infty e^{-xt} w(x) dx$$ for some non-negative weight function $w$ that is continuous on $(0,\infty)$ and satisfies $w(x) = \mathcal{O}(x^{-\gamma})$, $x\to 0$, and $w(x) = \mathcal{O}(x^{-\delta})$, $x\to\infty$. Furthermore, we define $A_{\gamma,\delta} \coloneqq \left(\frac{1}{\delta-1/2} + \frac{1}{2-\gamma}\right)^{1/2}.$
\end{definition}

\begin{theorem}\label{thm:TheGrandAKTheoremGeneral}
Let $x_0\in\R^d$, and let $b:\R^d\to\R$ and $\sigma:\R^{d\times d}\to\R^{d\times d}$ be globally Lipschitz continuous functions. Let $G$ be of type $(\gamma,\delta)$ with $1/2<\delta<3/2$. Let $X$ be the solution of \eqref{eqn:SVE}, and let $\hat{X}$ be the solution of \eqref{eqn:ApproximatedSVE} for the choice of nodes and weights in Appendix \ref{sec:ProofOfTheoremGeneral}. Then, $$\E\big|X_T - \hat{X}_T\big|^2 \le CN^{3/2-(\gamma\land\delta)}\exp\bigg(-\frac{2\alpha}{A_{\gamma,\delta}}\sqrt{N}\bigg),$$ where $C$ does not depend on $N$, and $\alpha$ is chosen as in Theorem \ref{thm:TheGrandAKTheorem}.
\end{theorem}

\subsection{Optimizing over constants}

As mentioned in Remark \ref{rmk:Thisw0NotOptimal}, we recommend choosing a weight $w_0$ different from the one stated in Theorem \ref{thm:TheGrandAKTheorem}. We will now analyse what the optimal choice of $w_0$ is. For this, we need the following exact error representation, which follows immediately by expanding the square.

\begin{proposition}\label{prop:fBmErrorRepresentation}
Consider $\hat{G}(t) = \sum_{i=0}^N w_i e^{-x_it},$ where $x_0 = 0$ and $x_i > 0$ for $i=1,\dots,N$. Then,
\begin{align}
\int_0^T \big|G(t) - \hat{G}(t)\big|^2 dt &= \frac{T^{2H}}{2H\Gamma(H+1/2)^2} + w_0^2 T + 2w_0 \sum_{i=1}^N \frac{w_i}{x_i} \Big(1-e^{-x_iT}\Big) \nonumber\\
&\qquad + \sum_{i,j=1}^N \frac{w_iw_j}{x_i + x_j} \Big(1-e^{-(x_i + x_j)T}\Big) \label{eqn:fBmErrorRepresentation}\\ 
&\qquad - \frac{2w_0T^{H+1/2}}{\Gamma(H+3/2)} - \frac{2}{\Gamma(H+1/2)}\sum_{i=1}^N \frac{w_i}{x_i^{H+1/2}} \int_0^{x_iT} t^{H-1/2} e^{-t} dt. \nonumber
\end{align}
\end{proposition}


\begin{remark}
The expression in \eqref{eqn:fBmErrorRepresentation} is almost explicit, except for the integral on the right hand side. This is however an incomplete gamma function, which can be computed efficiently.
\end{remark}

\begin{remark}\label{rmk:BetterW0}
Note that the right hand side of \eqref{eqn:fBmErrorRepresentation} is a quadratic polynomial in $w_0$. This polynomial can easily be minimized in closed form. Indeed, we propose using the $w_0$ minimizing equation \eqref{eqn:fBmErrorRepresentation} instead of the $w_0$ chosen in Theorem \ref{thm:TheGrandAKTheorem} or Theorem \ref{thm:TheGrandAKTheoremWithConstants}. This results in a slight improvement in the error.
\end{remark}

In Theorem \ref{thm:TheGrandAKTheorem} we have chosen Gaussian rules of type $(H, N, \alpha, \beta, a, b)$, where we have set $a=b=1$. Our goal is now to choose different $a$ and $b$ that improve the constants and the polynomial rate in Theorem \ref{thm:TheGrandAKTheorem}. We get Theorem \ref{thm:TheGrandAKTheoremWithConstants} below. The values of $a$ and $b$ in this theorem are chosen as the solution to some optimization problem. Since the proof of this theorem is somewhat technical and not very enlightening, we defer it to Appendix \ref{sec:ProofOfTheoremWithConstants}.

\begin{theorem}\label{thm:TheGrandAKTheoremWithConstants}
Let $x_0\in\R^d$, and let $b:\R^d\to\R$ and $\sigma:\R^{d\times d}\to\R^{d\times d}$ be globally Lipschitz continuous functions. Let $X$ be the solution of \eqref{eqn:SVE}, and let $\hat{X}$ be the solution of \eqref{eqn:ApproximatedSVE}, where we use the Gaussian rule of type $(H,N,\alpha,\beta,a,b)$ with $\alpha \coloneqq 1.06418$ and $\beta \coloneqq 0.4275,$
\begin{align*}
a &= T^{-1}\bigg(\Big(\frac{9-6H}{2H}\Big)^{\frac{e^{\alpha\beta}}{8(e^{\alpha\beta}-1)}}\Big(\frac{5\pi^3}{768}e^{\alpha\beta}\Big(e^{\alpha\beta}-1\Big)\frac{A^{2-2H}(3-2H)}{\beta^{2-2H}H}N^{1-H}\Big)^{2H}\bigg)^\gamma,\\
b &= T^{-1}\bigg(\Big(\frac{9-6H}{2H}\Big)^{\frac{e^{\alpha\beta}}{8(e^{\alpha\beta}-1)}}\Big(\frac{5\pi^3}{1152}e^{\alpha\beta}\Big(e^{\alpha\beta}-1\Big)\frac{A^{2-2H}}{\beta^{2-2H}}N^{1-H}\Big)^{2H-3}\bigg)^\gamma,
\end{align*}
where $\gamma \coloneqq \left( \frac{3e^{\alpha\beta}}{8(e^{\alpha\beta}-1)} + 6H - 4H^2\right)^{-1}$ and we assume that $N\ge 2$.
Then, $$\E\big|X_T - \hat{X}_T\big|^2 \le C(N) \exp\bigg(-\frac{2\alpha}{A}\sqrt{N}\bigg),$$ and $C(N)$ can be given explicitly as 
\begin{multline*}
C(N) = Cc_H^2T^{2H}\bigg(\frac{1}{2H} + 8\frac{e^{\alpha\beta}-1}{e^{\alpha\beta}} + \frac{1}{3-2H} \bigg) \times\\
\times \bigg(\Big(\frac{3}{H}\Big)^{\frac{e^{\alpha\beta}(3-2H)}{8(e^{\alpha\beta}-1)}}\Big(\frac{5\pi^3}{384}e^{\alpha\beta}\Big(e^{\alpha\beta}-1\Big)\frac{A^{2-2H}}{\beta^{2-2H}H}N^{1-H}\Big)^{6H-4H^2}\Big(\frac{1}{3/2-H}\Big)^{\frac{e^{\alpha\beta}H}{4(e^{\alpha\beta}-1)}}\bigg)^\gamma,
\end{multline*}
where $C$ is the constant from Proposition \ref{prop:AKCorollary}.
\end{theorem}

\begin{remark}\label{rmk:IntransparentRate}
The polynomial order of $C(N)$ in $N$ is rather intransparent. Taking the maximum over $H\in(0,1/2)$, we get that for any $H$ in this interval, $$C(N) = O(N^{0.42}).$$
\end{remark}

\section{Examples}
\label{sec:examples}

We apply the kernel approximation to three different use cases: fractional Brownian motion of Riemann--Liouville type, the rough Bergomi model, and finally, the rough Heston model.

\subsection{The case of fractional Brownian motion}
\label{sec:case-fract-brown}

Suppose we want to simulate a modified Riemann--Liouville fBm, i.e., $X_t = \frac{1}{\Gamma(H+1/2)} \int_0^t (t-s)^{H-1/2} dW_s $. It is not hard to see that the constant $C$ appearing in Proposition~\ref{prop:AKCorollary} can be chosen as $C=1$. Hence, Theorem~\ref{thm:TheGrandAKTheoremWithConstants} holds with $C=1$. To get a better feeling for the resulting constants and rates, let us specialize to $H = 0.1$ and $T = 1$. Then we get the following result.

\begin{corollary}\label{cor:TheGrandAKTheoremfBm}
For $H=0.1$, consider the approximation with $N+1$ nodes, one of them being at $x_0 = 0$, suggested above. Write $N = nm$, where $m$ is the level of the Gaussian quadrature rule and $n$ is the number of intervals. Choose $$m\approx 0.1306\sqrt{N},\qquad n = N/m \approx 7.6568\sqrt{N},$$ and $$\xi_0 = 4.3679N^{0.1135}e^{-0.2322\sqrt{N}},\quad \xi_n = 0.1421N^{-1.5889}e^{3.2511\sqrt{N}}.$$

Then, we have $$\E\big|X_T - \hat{X}_T\big|^2 \le 33.6483N^{0.3178}e^{-0.6502\sqrt{N}}.$$
\end{corollary}

\subsection{The case of the rough Bergomi model}
\label{sec:rBergomi}

The following presentation of the rough Bergomi model is largely inspired by \cite{B2016}. The rough Bergomi model is given by the system

\begin{align}
S_t &= S_0 \exp\bigg(\int_0^t \sqrt{V_s} \big(\rho dW_s + \sqrt{1-\rho^2} dB_s\big) - \frac{1}{2}\int_0^t V_s ds\bigg), \label{eqn:rBergomi}\\
V_t &= V_0 \exp\bigg(\eta\sqrt{2H}\int_0^t (t-s)^{H-1/2}dW_s - \frac{\eta^2}{2}t^{2H}\bigg).\nonumber
\end{align}

The rough Bergomi model was first studied by Bayer, Friz and Gatheral in \cite{B2016}, and can be seen as a rough, non-Markovian generalization of the Bergomi model by Bergomi \cite{B2005}. 
Here, the Hurst parameter $H$ controls the decay of the term structure of volatility skew for very short expirations and is usually around $0.1$, $\rho$ is the correlation between the Brownian motion driving the volatility process $V$ and the stock price $S$, and is usually negative, and the product $\rho\eta$ sets the level of the ATM skew for longer expirations. More information on how to choose these parameters can be found in \cite{B2016}. 

We can apply the above Markovian approximation of stochastic Volterra processes to the process $V$, or, more precisely, $\log V$. What we really care about, however, is the simulation of the stock price $S$. Assuming w.l.o.g. that $S_0 = 1$, we can rewrite \eqref{eqn:rBergomi} as follows.

\begin{align}
S_t &= 1 + \int_0^t S_s \exp(X_s) dW_s, \label{eqn:rBergomiAlterantiveForm}\\
dX_t &= d\bigg(\frac{1}{2}\log(V_t)\bigg) = \frac{\eta\sqrt{H}}{\sqrt{2}}\int_0^t (t-s)^{H-1/2} d\hat{W}_s. \nonumber
\end{align}

Now it is apparent that $X$ is the solution of a stochastic Volterra process with $b \equiv 0$ and $\sigma \equiv \eta\sqrt{H/2}.$

We recall the following Lemma by Harms \cite{H2019}.

\begin{lemma}\label{lem:HarmsrBergomi}\cite[Lemma 3]{H2019}.
Let $X,\hat{X},S,\hat{S}:[0,T]\times\Omega\to\R$ be continuous stochastic processes with $X_0 = \hat{X}_0 = 0$ and $$S_t = 1 + \int_0^t S_s \exp(X_s) dW_s,\qquad \hat{S}_t = 1 + \int_0^t \hat{S}_s \exp(\hat{X}_s) dW_s,$$ and let $f:(0,\infty)\to\R$ be a function such that $f\circ \exp$ is Lipschitz continuous with Lipschitz constant $L$. Then,
\begin{align*}
\big|\E f(S_T) - \E f(\hat{S}_T)\big| &\le L (\sqrt{T}+6) \left(\mathbb{E} \sup_{t\in[0,T]} \abs{\exp(2|X_t|) + \exp(2|\hat{X}_t|)}^2 \right)^{1/2}\\
&\qquad \times \left(\int_0^T \E\abs{X_t-\hat{X}_t}^2 dt\right)^{1/2}.
\end{align*}
\end{lemma}

The following corollary is immediate.

\begin{corollary}
Let $f:(0,\infty)\to\R$ be a function such that $f\circ\exp$ is Lipschitz continuous. Let $(S,X)$ be the solution of the stochastic Volterra equation \eqref{eqn:rBergomiAlterantiveForm}, and let $(\hat{S}, \hat{X})$ be the solution to the Markovian approximation using the point sets proposed in Theorem \ref{thm:TheGrandAKTheoremWithConstants} with $N+1$ points. Then, for some $C$ that can be chosen independent of $N$, we have $$\big|\E f(S_T) - \E f(\hat{S}_T)\big| \le C N^{0.21} \exp\bigg(-\frac{\alpha}{A}\sqrt{N}\bigg),$$ with $\alpha$ and $A$ as in Theorem \ref{thm:TheGrandAKTheoremWithConstants}.
\end{corollary}

\begin{proof}
Note that $X$ is, essentially, fractional Brownian motion, and that $\hat{X}$ is an Ornstein-Uhlenbeck process. Hence, exponential moments of $\sup_t|X_t|$ and $\sup_t|\hat{X}_t|$ exist, and we have by Lemma \ref{lem:HarmsrBergomi} that $$\big|\E f(S_T) - \E f(\hat{S}_T)\big| \le C_1 \|X-\hat{X}\|_{L^2([0,T]\times\Omega)}.$$ Now, by Theorem \ref{thm:TheGrandAKTheoremWithConstants}, (where we note that the constant $C(N)$ in the bound of the theorem can be chosen independent of $t$),
\begin{align*}
\|X-\hat{X}\|^2_{L^2([0,T]\times\Omega)} &= \int_0^T \E\big|X_t - \hat{X}_t\big|^2 dt\\
&\le \int_0^T C(N) \exp\bigg(-\frac{2\alpha}{A}\sqrt{N}\bigg) dt\\
&\le C_2 N^{0.42} \exp\bigg(-\frac{2\alpha}{A}\sqrt{N}\bigg).
\end{align*}
In the last step, we used Remark \ref{rmk:IntransparentRate}. Combining these estimates and taking the square root gives the result.
\end{proof}

\begin{remark}
  It should be noted that neither $S$ nor $\hat{S}$ can be simulated exactly, both require time discretization of the SDE~\eqref{eqn:rBergomiAlterantiveForm} defining $S$ or $\hat{S}$, respectively. So far, a proper weak error analysis of the Euler scheme for $S$ is missing -- but see \cite{bayer2020weak} for partial results, indicating weak rate $1/2 + H$. On the other hand, $(\hat{S}, \hat{X})$ is a standard diffusion process. Hence, an Euler discretization of $\hat{S}$ -- obviously, $\hat{X}$ does not need to be discretized further --  converges with weak rate $1$, for any fixed $N$. This does not contradict the weak rate $1/2 + H$ reported in \cite{bayer2020weak}, as the constant for the weak error will depend on $N$ as well as $H$, and might explode as $N \to \infty$.
\end{remark}

\subsection{The case of the rough Heston model}
\label{sec:rHeston}

The following presentation of the rough Heston model is largely inspired by \cite{A2019}. The rough Heston model is given by
\begin{align}
dS_t &= S_t \sqrt{V_t} dW_t, \label{eqn:HestonStock}\\
V_t &= V_0 + \int_0^t G(t-s) \big(\theta(s) -\lambda V_s\big) ds + \int_0^t G(t-s) \nu\sqrt{V_s} dB_s, \label{eqn:HestonVariance}
\end{align}
where $(W,B)$ is a two-dimensional correlated Brownian motion with correlation $\rho\in[-1,1]$, and $G$ is the usual rough kernel for some $H\in(0,1/2)$. Moreover, $\theta(s)$ is a deterministic, continuous function satisfying some conditions, as explained in \cite[Definition 2.1]{E2018}. This time-dependent mean-reversion level is used to fit the forward variance curve $(\E V_t)_{t\le T}$, as explained in \cite[Section 2]{A2019}. Also in \cite[Section 2]{A2019}, the existence of a weak non-negative solution to a generalization of \eqref{eqn:HestonVariance} with Hölder-regularity $H-\eps$ for all $\eps>0$ is shown.

Note that due to the square root in \eqref{eqn:HestonVariance}, the rough Heston model does not satisfy the Lipschitz assumptions of our framework. However, $(S,V)$ turns out to be an affine process. More precisely, as was shown in \cite{A2019b, E2019, E2018, E2019b}, the characteristic function of $\log(S_t/S_0)$ can be given in terms of the solution of a fractional Riccati equation, i.e. $$\E\exp\big(z\log(S_t/S_0)\big) = \exp\bigg(\int_0^t F(z, \psi(t-s, z)) g(s)ds\bigg).$$ Here, $$g(t) = V_0 + \int_0^t G(t-s) \theta(s) ds,$$ and $\psi(., z)$ is the unique continuous solution to the fractional Riccati equation
\begin{equation}\label{eqn:FractionalRiccatiEquation}
\psi(t, z) = \int_0^t G(t-s) F(z, \psi(s, z))ds,\qquad t\in[0,T],
\end{equation}
where
\begin{equation*}
F(z,x) \coloneqq \frac{1}{2}(z^2-z) + (\rho\nu z - \lambda)x + \frac{\nu^2}{2}x.
\end{equation*}
While equation \eqref{eqn:FractionalRiccatiEquation} cannot be solved explicitly, it can be solved numerically using the Adams scheme developed in \cite{D2002, D2004, D1998, E2019b}.

Now, we again replace the kernel $G$ by an approximation $\hat{G}$ of the type \eqref{eqn:ApproximatedKernel}. We denote the solution of the corresponding ordinary stochastic differential equations similar to equations \eqref{eqn:HestonStock} and \eqref{eqn:HestonVariance} by $(\hat{S},\hat{V})$. As was noted in \cite[Section 4.1]{A2019}, we again have $$\E\exp\big(z\log(\hat{S}_t/S_0)\big) = \exp\bigg(\int_0^t F(z, \hat{\psi}(t-s, z)) \hat{g}(s)ds\bigg).$$ Here, $$\hat{g}(t) = V_0 + \int_0^t \hat{G}(t-s) \theta(s) ds,$$ and $\hat{\psi}(., z)$ is the unique continuous solution to the Riccati equation $$\hat{\psi}(t, z) = \int_0^t \hat{G}(t-s) F(z, \hat{\psi}(s, z))ds,\qquad t\in[0,T].$$ Also in \cite[Section 4.1]{A2019}, it is illustrated that $\hat{\psi}$ is given as the solution of an $N$-dimensional system of ordinary Riccati equations, which can be solved numerically by usual numerical integrators for ODEs. Abi Jaber and El Euch were able to show the following two results in \cite{A2019}.

\begin{proposition}\label{prop:AbiJaberCharacteristicFunction}\cite[Theorem 4.1]{A2019}
There exists a constant $C>0$ such that for all $a\in[0,1]$, $b\in\R$ and $N\ge 1$, we have $$\sup_{t\in[0,T]} \big|\hat{\psi}(t, a+ib) - \psi(t, a+ib)\big| \le C(1+b^4) \int_0^T |\hat{G}(s) - G(s)| ds.$$
\end{proposition}

\begin{proposition}\label{prop:AbiJaberCallEstimates}\cite[Proposition 4.3]{A2019}
Let $C(k,T)$ denote the price of the call option in the rough Heston model with maturity $T>0$ and log-moneyness $k\in\R$. We denote by $\hat{C}(k,T)$ the price in the Markovian approximation of the rough Heston model. If $|\rho|<1$, then there exists a constant $c>0$ such that $$\big| C(k,T) - \hat{C}(k,T) \big| \le c \int_0^T \big|G(t) - \hat{G}(t) \big| dt.$$
\end{proposition}

The following two corollaries are immediate. We only prove the first one, as the proof for the second corollary is completely analogous.

\begin{corollary}
Assume that we take the point set proposed in Theorem \ref{thm:TheGrandAKTheoremWithConstants} with $N+1$ points. Then, there exists a constant $C>0$ such that for all $a\in[0,1]$, $b\in\R$ and $N\ge 0$, we have $$\sup_{t\in[0,T]} \big|\hat{\psi}(t, a+ib) - \psi(t, a+ib)\big| \le C(1+b^4)N^{0.21} \exp\bigg(-\frac{\alpha}{A}\sqrt{N}\bigg).$$
\end{corollary}

\begin{proof}
Using Proposition \ref{prop:AbiJaberCharacteristicFunction}, Jensen's inequality, Theorem \ref{thm:TheGrandAKTheoremWithConstants} and Remark \ref{rmk:IntransparentRate}, we get
\begin{align*}
\sup_{t\in[0,T]} \big|\hat{\psi}(t, a+ib) - \psi(t, a+ib)\big| &\le C_1(1+b^4) \int_0^T |\hat{G}(s) - G(s)| ds\\
&\le C_1 \sqrt{T}(1+b^4) \bigg(\int_0^T |\hat{G}(s) - G(s)|^2 ds\bigg)^{1/2}\\
&\le C_1 \sqrt{T}(1+b^4) \bigg(C(N) \exp\bigg(-\frac{2\alpha}{A}\sqrt{N}\bigg)\bigg)^{1/2}\\
&\le C (1+b^4) N^{0.21} \exp\bigg(-\frac{\alpha}{A}\sqrt{N}\bigg).\qedhere
\end{align*}
\end{proof}

\begin{corollary}\label{cor:rHestonCallErrorBound}
Let $C(k,T)$ denote the price of the call option in the rough Heston model with maturity $T>0$ and log-moneyness $k\in\R$. We denote by $\hat{C}(k,T)$ the price in the Markovian approximation of the rough Heston model, where we used $N+1$ points as suggested by Theorem \ref{thm:TheGrandAKTheoremWithConstants}. If $|\rho|<1$, then there exists a constant $c>0$ that can be chosen independent of $N$, such that $$\big| C(k,T) - \hat{C}(k,T) \big| \le c N^{0.21} \exp\bigg(-\frac{\alpha}{A}\sqrt{N}\bigg).$$
\end{corollary}


The main advantage of using the Markovian approximations above instead of the actual rough Heston model is of course that we have to solve an ordinary (multidimensional) Riccati equation instead of a fractional Riccati equation, which can usually be done faster. But using an ordinary Riccati equation instead of a fractional one has the additional advantage that the rate of convergence of the discretization of the Riccati equation is higher. Indeed, it was proved in \cite{li2009fractional} that the rate of convergence of the Adams scheme that is usually used for solving the fractional Riccati equation is $O(\Delta t)$, where $\Delta t$ is the step size. At the same time, we can solve the ordinary Riccati equation with a simple predictor corrector method, for example, which has rate of convergence $O(\Delta t^2)$.

\section{Numerics}
\label{sec:numerics}

\subsection{Optimal choices of $m$, $\xi_0$ and $\xi_n$ for $H=0.1$}\label{sec:OptimalChoices}

Suppose we take the approach as before to choose $N+1$ points, one of them at $x_0 = 0$ with the optimal weight given by Remark \ref{rmk:BetterW0}, and the other $N$ points are chosen by applying Gaussian quadrature of level $m$ on $n$ intervals that are geometrically spaced on $[\xi_0,\xi_n]$, and such that $mn\approx N$. Theorems \ref{thm:TheGrandAKTheorem} and \ref{thm:TheGrandAKTheoremWithConstants} give us possible choices of $m$, $\xi_0$ and $\xi_n$, but what are the optimal choices? Using the exact error representation given by Proposition \ref{prop:fBmErrorRepresentation}, we can minimize this error over $\xi_0$ and $\xi_n$ for each $m$. More precisely, for each $m=1,\dots,10$, we try to minimize the right-hand side of Proposition \ref{prop:fBmErrorRepresentation} as a function of $\xi_0$ and $\xi_n$, where for each step in the optimization procedure, we compute the nodes $(x_i)$ and the weights $(w_i)$ corresponding to the current choice of $\xi_0$ and $\xi_1$ according to our definition of Gaussian rules in the beginning of Section \ref{sec:SuperpolynomialRateOfConvergence}. One may of course ask why we do not directly optimize over the $(x_i)$ and $(w_i)$. The reason is that the resulting optimization problem would be very high-dimensional, and non-convex. Conversely, using Gaussian rules, we only have to optimize over 2 (or 3 if we include $m$) parameters. Although we are only optimizing over two parameters, the problem is still quite delicate, and we may not have gotten the exact global optima. Nonetheless, we expect our results to be very close to the optima. Figure \ref{fig:FBmErrorDifferentLevels} illustrates the optimized $L^2$-errors (i.e. where we have taken the square root of the expression in Proposition \ref{prop:fBmErrorRepresentation}) for the choice $H=0.1$ and $T=1$. We can see how, the larger $N$ gets, we should also choose larger $m$ to get the best results.

\begin{figure}[!htbp]
\centering
\includegraphics[width=0.7\columnwidth]{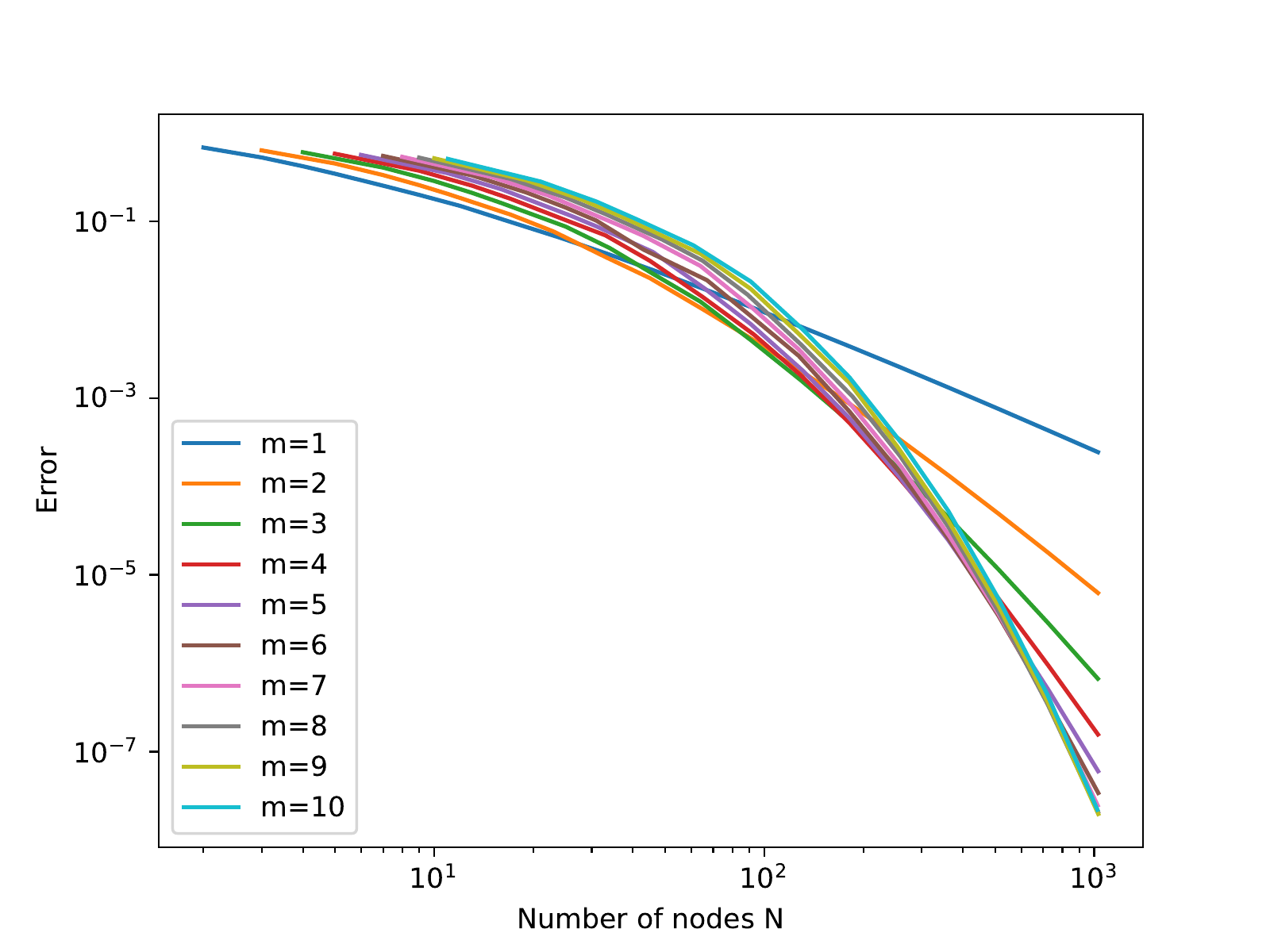}
\caption{$L^2$-error for the approximation of fractional Brownian motion with $H=0.1$ and $T=1$ for different quadrature levels $m$ as a function of the number of nodes $N$, using the optimal values for $\xi_0$ and $\xi_n$}
\label{fig:FBmErrorDifferentLevels}
\end{figure}

In Table \ref{tab:TheoremComparedToOptimum} we compare the point sets for various different parameter choices.

\begin{table}[!htbp]
\centering
\resizebox{\textwidth}{!}{\begin{tabular}{c|c|c|c|c|c|c|c|c|c|c|c|c|c|c}
 & \multicolumn{6}{|c|}{Theorem} & \multicolumn{3}{|c|}{Optimal $\xi_0,\ \xi_n$} & \multicolumn{5}{|c}{Optimal $m,\ \xi_0,\ \xi_n$}\\ \hline
$N$ & $m$ & $n$ & $-\log\xi_0$ & $\log\xi_n$ & Error & Bound & $-\log\xi_0$ & $\log\xi_n$ & Error & $m$ & $n$ & $-\log\xi_0$ & $\log\xi_n$ & Error \\ \hline
1 & 1 & 1 & -1.2421 & 1.2999 & 0.996490 & 4.190757 & 83.372 & 4.8778 & 0.683687 & 1 & 1 & 83.372 & 4.8778 & 0.683687 \\
2 & 1 & 2 & -1.2246 & 1.5452 & 0.975001 & 4.089250 & 1.4621 & 9.1800 & 0.528237 & 1 & 2 & 1.4621 & 9.1800 & 0.528237 \\
4 & 1 & 4 & -1.1672 & 2.3483 & 0.899764 & 3.773728 & 0.7776 & 14.455 & 0.346109 & 1 & 4 & 0.7776 & 14.455 & 0.346109 \\
8 & 1 & 8 & -1.0535 & 3.9403 & 0.757286 & 3.218395 & 0.5037 & 21.394 & 0.199291 & 1 & 8 & 0.5037 & 21.394 & 0.199291 \\
16 & 1 & 16 & -0.8602 & 6.6478 & 0.571030 & 2.455046 & 1.6463 & 28.971 & 0.098625 & 1 & 16 & 1.6463 & 28.971 & 0.098625 \\
32 & 1 & 32 & -0.5541 & 10.933 & 0.372303 & 1.599424 & 2.3096 & 37.865 & 0.043699 & 2 & 16 & 1.8629 & 36.893 & 0.039571 \\
64 & 2 & 32 & -0.0887 & 17.450 & 0.195570 & 0.833625 & 2.7007 & 51.739 & 0.010167 & 2 & 32 & 2.7007 & 51.739 & 0.010167 \\
128 & 2 & 64 & 0.6021 & 27.121 & 0.075222 & 0.316916 & 4.4629 & 68.195 & 0.002037 & 3 & 43 & 3.9939 & 70.067 & 0.001559 \\
256 & 3 & 85 & 1.6115 & 41.256 & 0.018481 & 0.077112 & 6.5970 & 93.266 & 0.000158 & 4 & 64 & 6.2656 & 95.048 & 0.000123 \\
512 & 3 & 171 & 3.0718 & 61.701 & 0.002405 & 0.009982 & 8.2385 & 120.50 & 1.14e-05 & 6 & 85 & 9.4066 & 130.22 & 3.50e-06 \\
1024 & 5 & 205 & 5.1694 & 91.071 & 0.000128 & 0.000529 & 12.162 & 172.37 & 6.03e-08 & 9 & 114 & 13.029 & 180.30 & 1.98e-08 
\end{tabular}}
\caption{$L^2$-approximation errors of fractional Brownian motion with $H=0.1$ and $T=1$ for various choices of $N$, $m$, $n$, $\xi_0$ and $\xi_n$. The column labeled ``Theorem'' denotes the values of $m$, $n$, $-\log\xi_0$ and $\log\xi_n$ suggested by Corollary \ref{cor:TheGrandAKTheoremfBm}. The error column denotes the actual error for these choices, computed using Proposition \ref{prop:fBmErrorRepresentation}. The bound column gives the error bound in Corollary \ref{cor:TheGrandAKTheoremfBm}. The column labelled ``Optimal $\xi_0$, $\xi_n$'' uses the same values for $m$ and $n$ as suggested by Corollary \ref{cor:TheGrandAKTheoremfBm}, but takes the corresponding optimal values for $\xi_0$ and $\xi_n$ that can be found using an optimization algorithm. The error column again shows the errors for these choices. The column labelled ``Optimal $m$, $\xi_0$, $\xi_n$'' takes the optimal value for $m$ and the corresponding value for $n$ given $N$, and then also the optimal $\xi_0$ and $\xi_n$ as calculated by optimization. Finally, again the error is given in the error column.}
\label{tab:TheoremComparedToOptimum}
\end{table}

Note that the optimal choices of $-\log\xi_0$ and $\log\xi_n$ are larger than the choices made in the theorem. This indicates that the theorem is overestimating the quadrature errors, which leads to artificially small intervals $[\xi_i,\xi_{i+1}]$. Hence, as indicated in Remark \ref{rmk:RoundmUp}, we recommend always rounding $m$ up when applying Theorem \ref{thm:TheGrandAKTheorem} or Theorem \ref{thm:TheGrandAKTheoremWithConstants}, to slightly improve the results. Also, a comparison of the columns ``Optimal $\xi_0$, $\xi_n$'' and ``Optimal $m$, $\xi_0$, $\xi_n$'' shows that the correct choice of $m$ is not extremely crucial, as long as it is of the right order of magnitude.

Figure \ref{fig:FBmErrorComparison} illustrates the errors reported in Table \ref{tab:TheoremComparedToOptimum}.

\begin{figure}[!htbp]
\centering
\includegraphics[width=0.7\columnwidth]{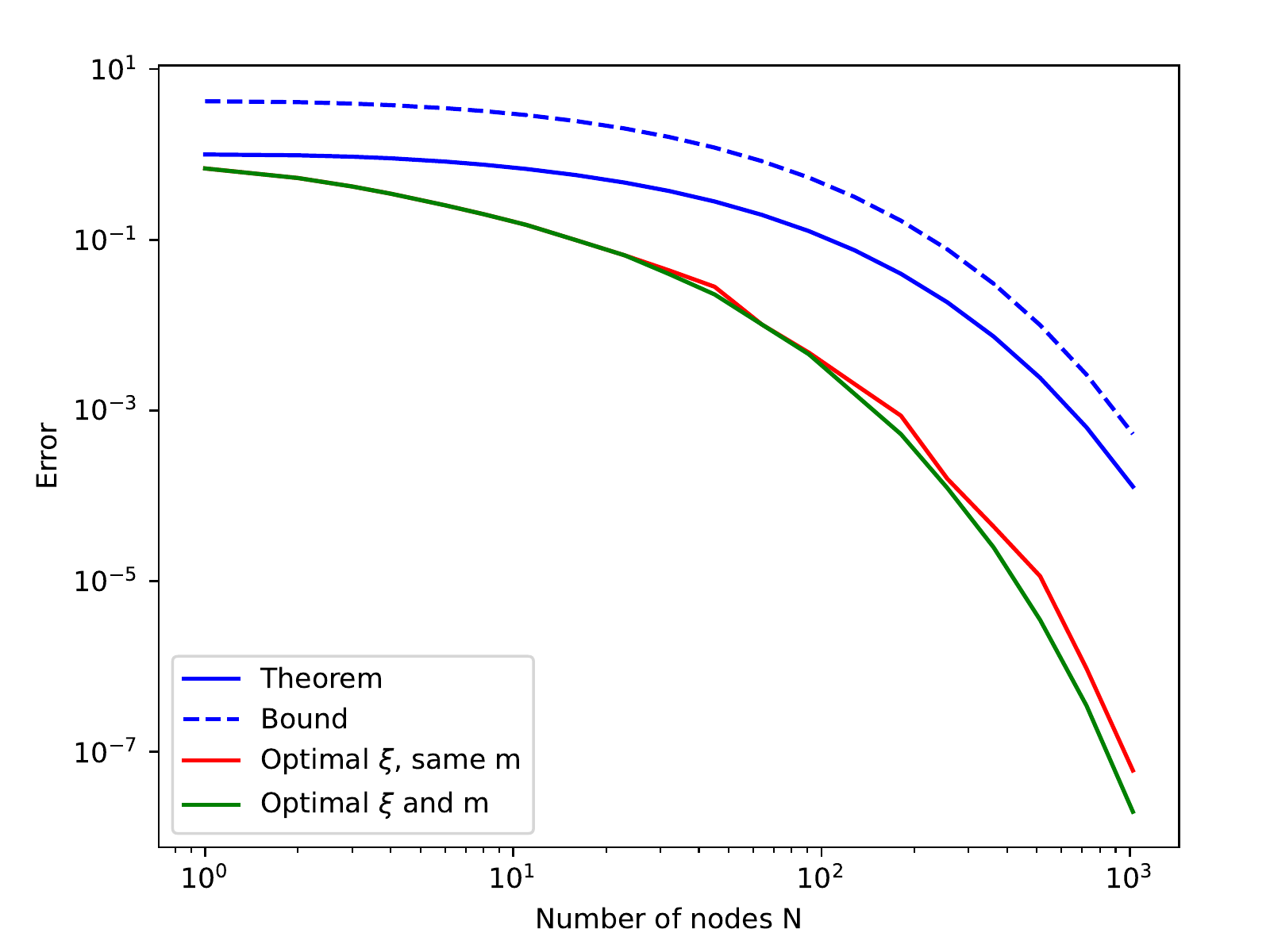}
\caption{$L^2$-approximation errors for fractional Brownian with $H=0.1$ and $T=1$. ``Theorem'' is the true error using the values of $m$, $n$, $\xi_0$ and $\xi_n$ specified in Corollary \ref{cor:TheGrandAKTheoremfBm}, while ``Bound'' is the corresponding error bound of Corollary \ref{cor:TheGrandAKTheoremfBm}. ``Optimal $\xi$, same $m$'' uses the same $m$ as the one suggested in Corollary \ref{cor:TheGrandAKTheoremfBm}, but the corresponding optimized values for $\xi_0$ and $\xi_n$. ``Optimal $\xi$ and $m$'' uses both the optimal values for $m$ and for $\xi_0$ and $\xi_n$.}
\label{fig:FBmErrorComparison}
\end{figure}

\subsection{Learning a better rate of convergence}
\label{sec:LearningBetterRate}

While Theorem \ref{thm:TheGrandAKTheoremWithConstants} produced decent results, Figure \ref{fig:FBmErrorComparison} illustrates further potential improvements. Using optimization to get the optimal values of $m$, $\xi_0$ and $\xi_n$ for any $N$ and $H$, and then applying regression, we get the following approximations of the optimal parameter values as a function of $N$.

\begin{align*}
\xi_0 &= 0.65T^{-1}e^{3.1H}\exp\bigg(-\frac{1.8}{(3/2-H)A}\sqrt{N}\bigg),\\
\xi_n &= T^{-1}e^{3H^{-0.4}}\exp\bigg(\frac{1.8}{HA}\sqrt{N}\bigg),\\
m &= \frac{0.9}{A}\sqrt{N}.
\end{align*}

We even get a corresponding error estimate, given by $$\bigg(\int_0^T \big|G(t) - \hat{G}(t)\big|^2\bigg)^{1/2} \approx T^He^{0.065H^{-1.1}}\exp\bigg(-\frac{1.8}{A}\sqrt{N}\bigg)$$ where ``$\approx$'' has the usual meaning of ``is approximately equal to''. Some further explanation on how we achieved these results can be found in Appendix \ref{sec:AppendixLearningBetterRate}.

Note that when applying the values of $\xi_0,$ $\xi_n$ and $m$ suggested above, we do not recommend always rounding $m$ up, in contrast to Remark \ref{rmk:RoundmUp}. Instead, we recommend rounding $m$ to the closest positive integer.

For the case $H=0.1$ and $T=1$, we compare the different approximations and estimates in Table \ref{tab:FBmErrorComparisonWithExistingMethods}. The original method of Alfonsi and Kebaier is the one taken from \cite[Corollary 3.1]{A2021}, while the improved one is taken from \cite[Table 6, left column]{A2021}, where we remark that both versions exhibit a polynomial convergence rate of $0.8H$. Harms is the method in \cite{H2019} with Gaussian quadrature level $m=1$ and $m=10$, respectively. Theorem refers to the point set given in Corollary \ref{cor:TheGrandAKTheoremfBm}, Numerical estimates refers to the numerically inferred point set above, and Optimum is the optimized error we achieved before.

\begin{table}[!htbp]
\centering
\resizebox{\textwidth}{!}{\begin{tabular}{c|c|c|c|c|c|c|c|c|c}
 & \multicolumn{2}{c|}{Alfonsi, Kebaier} & \multicolumn{2}{c|}{Harms} & \multicolumn{2}{|c|}{Theorem} & \multicolumn{2}{|c|}{Numerical estimates} & \multicolumn{1}{|c}{Optimum}\\ \hline
$N$ & Original & Improved & $m=1$ & $m=10$ & Error & Bound & Error & Bound & Error \\ \hline
1 & 1.093956 &&&& 0.996490 & 4.190757 & 0.917761 & 1.307860 & 0.683687 \\
2 & 1.038641 & 0.656022 & 1.408506 && 0.975001 & 4.089250 & 0.697745 & 1.041449 & 0.528237 \\
4 & 0.987654 & 0.471510 & 1.318878 && 0.899764 & 3.773728 & 0.389907 & 0.754638 & 0.346109 \\
8 & 0.940801 & 0.290318 & 1.235692 && 0.757286 & 3.218395 & 0.211681 & 0.478511 & 0.199291 \\
16 & 0.897258 & 0.151103 & 1.160442 & 0.836713 & 0.571030 & 2.455046 & 0.098789 & 0.251243 & 0.098625 \\
32 & 0.856277 & 0.068764 & 1.092909 & 0.632356 & 0.372303 & 1.599424 & 0.041534 & 0.101018 & 0.039571 \\
64 & 0.817406 & 0.028360 & 1.032075 & 0.398268 & 0.195570 & 0.833625 & 0.010345 & 0.027849 & 0.010167 \\
128 & 0.780402 & 0.010653 & 0.976870 & 0.238136 & 0.075222 & 0.316916 & 0.001611 & 0.004502 & 0.001559 \\
256 & 0.745117 & 0.003793 & 0.926378 & 0.150077 & 0.018481 & 0.077112 & 0.000124 & 0.000342 & 0.000123 \\
512 & 0.711447 & 0.001506 & 0.879856 & 0.095788 & 0.002405 & 0.009982 & 3.72e-06 & 8.94e-06 & 3.50e-06 \\
1024 & 0.679308 & 0.000933 & 0.836708 & 0.060342 & 0.000128 & 0.000529 & 2.24e-08 & 5.16e-08 & 1.98e-08
\end{tabular}}
\caption{$L^2$ errors for the approximation of fractional Brownian motion with $H=0.1$ and $T=1$ using different point sets. Since the point sets of Harms are only well-defined if $n\ge 2$, there are some gaps in the table. Also, the number $N$ of nodes is only approximate. For example, Harms' point set with $n=2$ and $m=10$ has 20 points, but is in the row $N=16$, as this the value closest to 20.}
\label{tab:FBmErrorComparisonWithExistingMethods}
\end{table}

Figure \ref{fig:FBmErrorComparisonWithExistingMethods} plots the errors reported in Table \ref{tab:FBmErrorComparisonWithExistingMethods}. Notice how close the errors of using the numerical estimates of the optimal choices of $m$, $\xi_0$ and $\xi_n$ that we got in this section are to the optimal errors.

\begin{figure}[!htbp]
\begin{center}
\includegraphics[width=0.7\columnwidth]{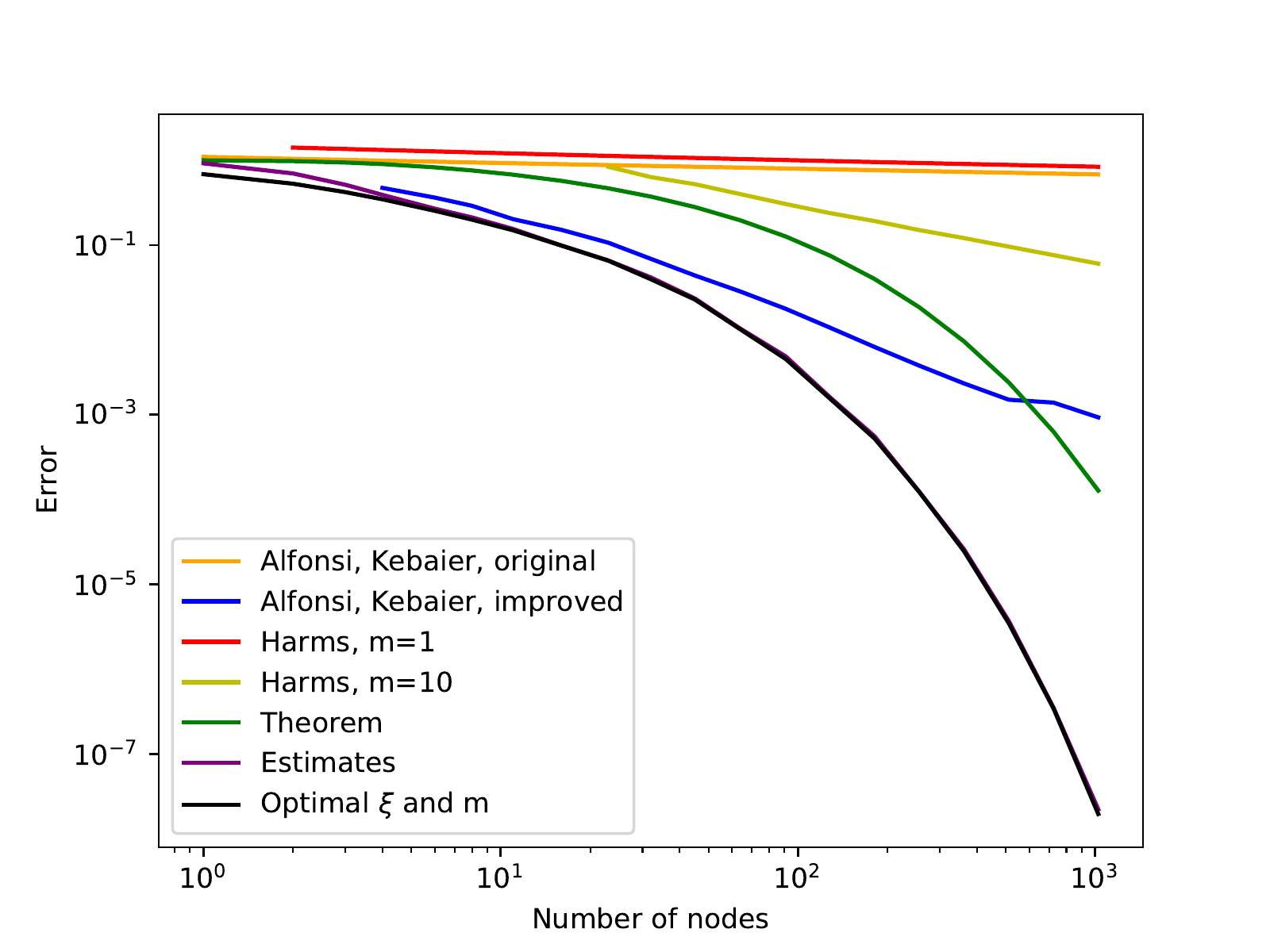}
\end{center}
\caption{$L^2$ errors for the approximation of fractional Brownian motion with $H=0.1$ and $T=1$ using different point sets. ``Estimates'' refers to the point sets obtained in this section.}
\label{fig:FBmErrorComparisonWithExistingMethods}
\end{figure}

Finally, in Table \ref{tab:HowManyPointsDoWeNeed} we give the number of points required to achieve a certain error in the $L^2$-approximation of fractional Brownian motion for different Hurst parameters and time $T=1$.

\begin{table}[!htbp]
\centering
\begin{tabular}{c|c|c|c|c}
$H$ & Error tolerance & Theorem bound & Theorem & Numerical estimates\\\hline
0.05 & 0.1 & 524 & 242 & 50\\
0.05 & 0.01 & 1093 & 661 & 161\\
0.1 & 0.1 & 230 & 108 & 16\\
0.1 & 0.01 & 512 & 325 & 66\\
0.2 & 0.1 & 90 & 36 & 4\\
0.2 & 0.01 & 228 & 139 & 26
\end{tabular}
\caption{Number of nodes needed to achieve a certain error for a given Hurst parameter. ``Theorem bound'' refers to the bound in Theorem \ref{thm:TheGrandAKTheoremWithConstants}, ``Theorem'' uses the actual errors of the point sets of Theorem \ref{thm:TheGrandAKTheoremWithConstants}, and ``Numerical estimates'' refers uses the point sets with the learned values for $m$, $\xi_0$ and $\xi_n$.}
\label{tab:HowManyPointsDoWeNeed}
\end{table}

\subsection{Implied volatility smile of the rough Bergomi model}\label{sec:rBergomiSmile}

We now apply our approximation to the rough Bergomi model, see Section~\ref{sec:rBergomi}. We compute the implied volatility using Monte Carlo simulation. We give the algorithm for simulating the Markovian approximation below. 

Recall the definition of the rough Bergomi model in \eqref{eqn:rBergomi}. Here, we replace the integral $$\int_0^t (t-s)^{H-1/2} dW_s$$ in the exponent of the volatility by our Markovian approximation. Also, note that the $t^{2H}$ in the exponent of the volatility corresponds to the square of the $L^2$-norm of the aforementioned stochastic integral. Hence, we also replace $t^{2H}$ by the square of the $L^2$-norm of the Markovian approximation. We obtain the system
\begin{align*}
S_t &= S_0\exp\left(\int_0^t \sqrt{V_s}\left(\rho dW_s + \sqrt{1-\rho^2}dB_s\right)-\frac{1}{2}\int_0^t V_s ds\right),\\
V_t &= V_0\exp\Bigg(\eta\sqrt{2H}\Gamma(H+1/2)\sum_{i=0}^N w_i \int_0^t e^{-x_i(t-s)} dW_s\\
&\qquad\qquad - 2\eta^2H\Gamma(H+1/2)^2\int_0^t \left(\sum_{i=0}^N w_i e^{-x_i(t-s)} \right)^2 ds\Bigg).
\end{align*}

Now we discretize this system in time. Let $0 = t_0 < \dots < t_m = T$ be a uniform partition with step size $\Delta t = T/m$. Initialize the $N+3$ components $\overline{S}$, $\overline{V}$, and $(\overline{V}^i)_{i=0}^N$ by $$\overline{S}_0 = S_0,\qquad \overline{V}_0 = V_0,\qquad \overline{V}^i_0 = 0.$$ Here, $V^i_t$ corresponds to the integral $$\int_0^t e^{-x_i(t-s)} dW_s.$$

If we have already simulated the system for time $t_j$, we then simulate the Gaussian vector $$\left(\int_{t_j}^{t_{j+1}} e^{-x_0(t_{j+1} - s)} dW_s,\dots, \int_{t_j}^{t_{j+1}} e^{-x_N(t_{j+1}-s)} dW_s, W_{t_{j+1}} - W_{t_j}\right).$$ If the grid was chosen uniformly in time, the distribution of this vector is of course independent of $j$, and we only have to compute a Cholesky decomposition once.

Then, we set
\begin{align*}
\overline{S}_{t_{j+1}} &= \overline{S}_{t_j} \exp\left(\sqrt{\overline{V}_{t_j}}\left(\rho (W_{t_{j+1}} - W_{t_j}) + \sqrt{1-\rho^2}(B_{t_{j+1}}-B_{t_j})\right) - \frac{\Delta t}{2}\overline{V}_{t_j}\right),\\
\overline{V}^i_{t_{j+1}} &= e^{-x_i\Delta t} \overline{V}^i_{t_j} + \int_{t_j}^{t_{j+1}} e^{-x_i(t_{j+1}-s)} dW_s,\\
\overline{V}_{t_{j+1}} &= V_0\exp\Bigg(\eta\sqrt{2H}\Gamma(H+1/2)\sum_{i=0}^N w_i \overline{V}^i_{t_{j+1}}\\
&\qquad\qquad - 2\eta^2H\Gamma(H+1/2)^2\int_0^t \left(\sum_{i=0}^N w_i e^{-x_i(t-s)} \right)^2 ds\Bigg).
\end{align*}

Here, the final deterministic integral can easily be computed in closed form.

Using the parameters $T=0.9$, $H=0.07$, $\eta = 1.9$, $\rho=-0.9$, $S_0 = 1$, $V_0 = 0.235^2$, and using 2000 time steps and $10^6$ samples for the Monte Carlo estimates, we get the implied volatility smiles shown in Figure \ref{fig:RBergomiSmiles}. The number of time steps and samples were chosen to ensure that the time discretization and Monte Carlo errors are sufficiently small. In the figure, the dashed black lines are the $95\%$ confidence interval of the volatility smile where the differential equations were only discretized in time, as suggested by Bayer, Friz and Gatheral in \cite{B2016}. The confidence intervals of the approximations have similar sizes and shapes, but were not drawn for better interpretability. Note how the smile generated for $N=16$ is already almost indistinguishable to the human eye from the true smile. We recall that by Lemma \ref{lem:HarmsrBergomi}, the errors in the smile are directly related to the $L^2$-errors of the approximation of $K$. Finally, we remark that the parameter choice we used was the one suggested in Section \ref{sec:LearningBetterRate}.

\begin{figure}[!htbp]
  \centering
    \begin{subfigure}[t]{0.49\textwidth}
        \centering
        \includegraphics[width=\textwidth]{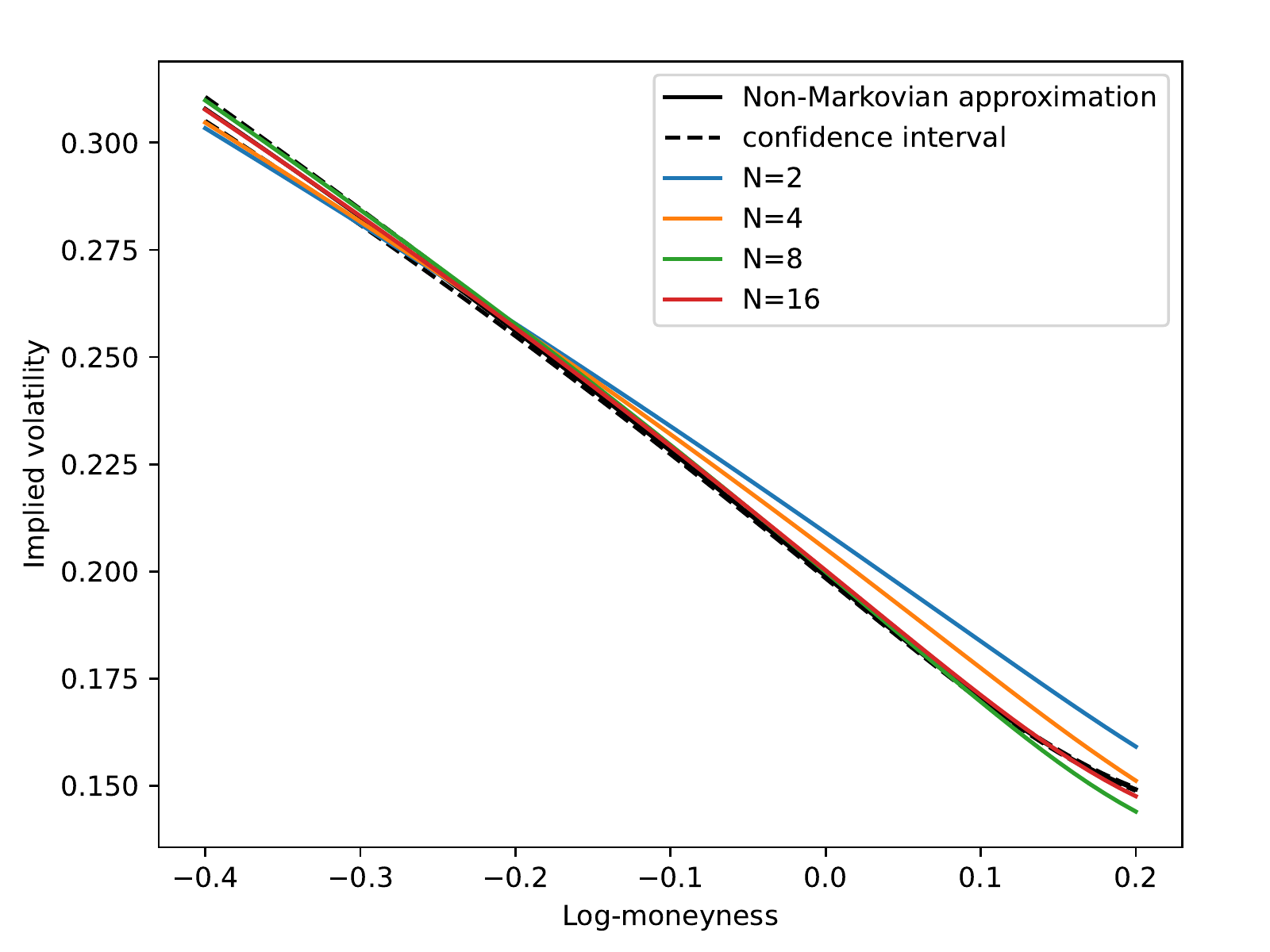}
        \caption{Approximations based on Theorem~\ref{thm:TheGrandAKTheoremWithConstants} for different $N$.}
        \label{fig:RBergomiSmiles}
    \end{subfigure}~
    \begin{subfigure}[t]{0.49\textwidth}
        \centering
        \includegraphics[width=\textwidth]{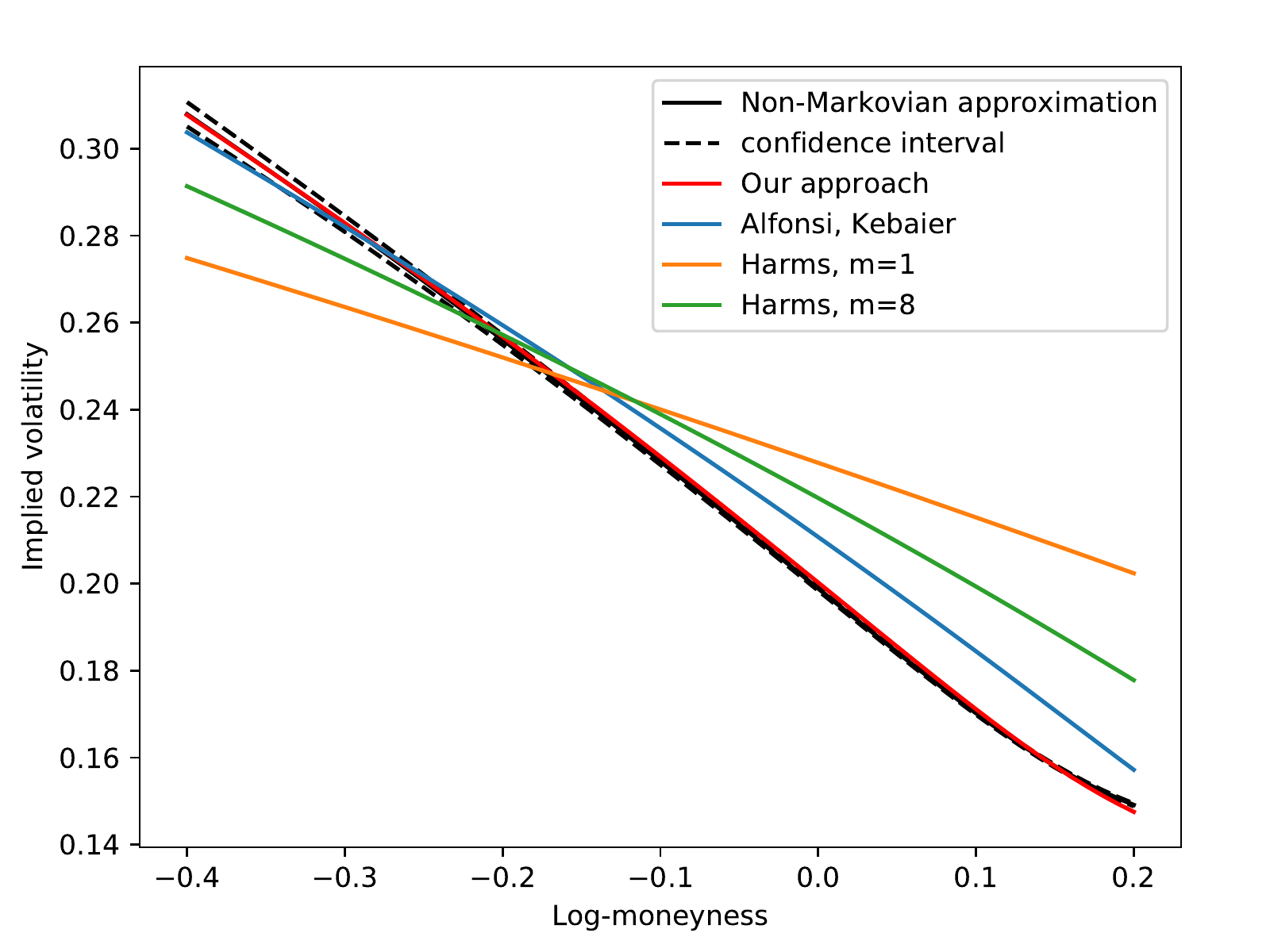}
        \caption{Comparison with other methods from the literature, $N = 16$.}
        \label{fig:RBergomiSmilesComparison}
    \end{subfigure}
  \caption{Implied volatility smiles on the rough Bergomi model for the approximation given in \cite{B2016} and the Markovian approximations.}
  \label{fig:RBergomi}
\end{figure}

The time in seconds it took to generate these smiles is given in Table \ref{tab:RBergomiSmilesTime}. Here, $N=\infty$ is the non-Markovian method described in \cite{B2016}, which can also be formally attained by taking $N\to\infty$.

\begin{table}[!htbp]
\centering
\begin{tabular}{c|c|c|c|c|c|c|c}
N & 1 & 2 & 4 & 8 & 16 & 32 & $\infty$ \\ \hline
Time & 78 & 123 & 230 & 490 & 1318 & 5563 & 2579 
\end{tabular}
\caption{Time in seconds to generate the implied volatility smiles for different values of $N$. Here, $N=\infty$ refers to the method in \cite{B2016}.}
\label{tab:RBergomiSmilesTime}
\end{table}

We can see in Figure \ref{fig:RBergomiSmiles} that the choice $N=16$ already produced good results. We can now compare our approximation with 16 nodes to the approximations given in \cite{A2021} and \cite{H2019} with an equal number of nodes. Since the approximations are all of the same nature, and only the choice of nodes and weights differs, keeping the number of nodes constant means keeping the computational cost constant. Hence, in Figure \ref{fig:RBergomiSmilesComparison} we see again the volatility smile generated using the non-Markovian approximation from \cite{B2016} including the $95\%$ confidence intervals, as well as the volatility smiles generated using our point set, the point set in \cite[Corollary 3.1]{A2021}, and the point set in \cite{H2019} with Gaussian quadrature levels $1$ and $8$.


\subsection{Implied volatility smile of the rough Heston model}\label{sec:rHestonSmile}

Finally, we apply our approximation to the rough Heston model, see Section~\ref{sec:rHeston}.
We choose the same parameters as in \cite[Section 4.2]{A2019}, i.e. $$\lambda = 0.3,\quad \rho=-0.7,\quad \nu=0.3,\quad H=0.1,\quad V_0 = 0.02,\quad \theta=0.02,\quad T=1,\quad S_0=1.$$ The implied volatility smiles are computed using Fourier inversion, in the spirit of Section \ref{sec:rHeston}. As described above, we need to solve (fractional or ordinary multidimensional) Riccati equations. For the fractional Riccati equation we use the Adams scheme, which is well explained in \cite[Section 5.1]{E2019b}. Essentially, this scheme is a predictor-corrector method.

We now explain how we solve the ordinary multidimensional Riccati equations. As explained in \cite[Section 4.1]{A2019}, we have to solve the system
\begin{align}
\widehat{\psi}(t,z) &= \sum_{i=1}^n w_i \psi^{x_i} (t,z), \nonumber\\
\partial_t \psi^{x_i}(t,z) &= -x_i\psi^{x_i}(t,z) + F(z,\widehat{\psi}(t,z)),\qquad \psi^{x_i}(0,z) = 0, \label{eqn:OrdinaryRiccati}
\end{align}
where $$F(z,x) = \frac{1}{2}(-z^2-iz) + \lambda(\rho\nu iz - 1) x + \frac{(\lambda\nu)^2}{2}x^2.$$ For comparability, we would also like to solve this system using a predictor-corrector method. However, the usual Euler method with the trapezoidal rule will not work here. The reason is that the mean reversions $x_i$ can be absurdly large, making any Euler approximation completely intractable. Hence, we only discretize the second summand in the right-hand side of \eqref{eqn:OrdinaryRiccati} and keep the first summand as it is.

More precisely, choose a step size $\Delta$ and set $t_k = k\Delta$. Given an approximation $\widetilde{\psi}(t_k,z)$ of $\widehat{\psi}(t_k,z)$, approximate $$\partial_t \psi^x(t,z) = -x\psi^x(t,z) + F(z,\widehat{\psi}(t,z))$$ on the interval $[t_k,t_{k+1}]$ by $$\partial_t \widetilde{\psi}^x(t,z) = x\widetilde{\psi}^x(t,z) + F(z, \widetilde{\psi}(t_k,z)).$$ This is now an ordinary, one-dimensional differential equation of the form $$\frac{d}{dt} y(t) = ay(t) + b.$$ The solution to this ODE is $$y(t_{k+1}) = \frac{b}{a}\Big(e^{a\Delta} -1\Big) + y(t_k)e^{a\Delta}.$$ Using this, we apply the following predictor-corrector method. Given $\widetilde{\psi}(t_k,z)$ and $\widetilde{\psi}^x(t_k,z)$, we set
\begin{align*}
\widetilde{\psi}^{x,P} &= \frac{F(z,\widetilde{\psi}(t_k,z))}{x}\Big(1-e^{-x\Delta}\Big) + \widetilde{\psi}^x(t_k,z) e^{-x\Delta},\\
\widetilde{\psi}^P &= \frac{1}{2}\bigg(\widetilde{\psi}(t_k,z) + \sum_{i=0}^N w_i \widetilde{\psi}^{x_i,P}\bigg),\\
\widetilde{\psi}^x(t_{k+1},z) &= \frac{F(z,\widetilde{\psi}^P)}{x}\Big(1-e^{-x\Delta}\Big) + \widetilde{\psi}^x(t_k,z) e^{-x\Delta},\\
\widetilde{\psi}(t_{k+1},z) &= \sum_{i=0}^N w_i \widetilde{\psi}^{x_i}.
\end{align*}

Having described our schemes for the solutions of the Riccati equations, we now describe our approximations of the Fourier inversion. To that end, recall the Fourier inversion formula
\begin{equation}\label{eqn:FourierInversionFormula}
\E f(X) = \frac{1}{2\pi} \int_{\R} \varphi_X(u+iR)\hat{f}(u+iR) du 
\end{equation}
for suitable functions $f$ and random variables $X$. Here, $\varphi_X(u) = \E e^{iuX}$ is the characteristic function of $X$, and $\hat{f}$ is the Fourier transform of $f$. Furthermore, $R\in\R$ is a parameter that has to be chosen suitably to ensure existence of $\varphi_X(u+iR)$ and some integrability conditions on $f$ and $\hat{f}$. We use this inversion formula with $$X=\log(S_T)/\log(S_0),\qquad \text{and}\qquad f(x) = (e^x-K)^+.$$ Moreover, we choose the parameter $R=2$. We then need to approximate the integral in \eqref{eqn:FourierInversionFormula}. To this end, we somewhat arbitrarily truncate the integral and restrict ourselves to $[-50,50]$. On this interval, we then use the trapeziodal rule on an equidistant grid with 10000 points. Finally, for both the fractional and the ordinary multidimensional Riccati equations we use 3000 time steps.

Figure \ref{fig:RHestonSmiles} shows the implied volatility smile of the rough Heston model under the above mentioned parameters. The black line is the non-Markovian scheme as described in \cite{E2019b}, while the other lines use our Markovian approximation for different numbers $N$ of nodes with the parameters given Section \ref{sec:LearningBetterRate}. Note how, even for $N=6$, our approximation is barely distinguishable from the exact model as approximated by the Adam scheme. We further compare our approximation with $N=16$ to the approximation given in \cite{A2019} with a varying numbers of  nodes, see Figure \ref{fig:RHestonSmilesComparison}.

\begin{figure}[!htbp]
  \centering
    \begin{subfigure}[t]{0.49\textwidth}
        \centering
        \includegraphics[width=\textwidth]{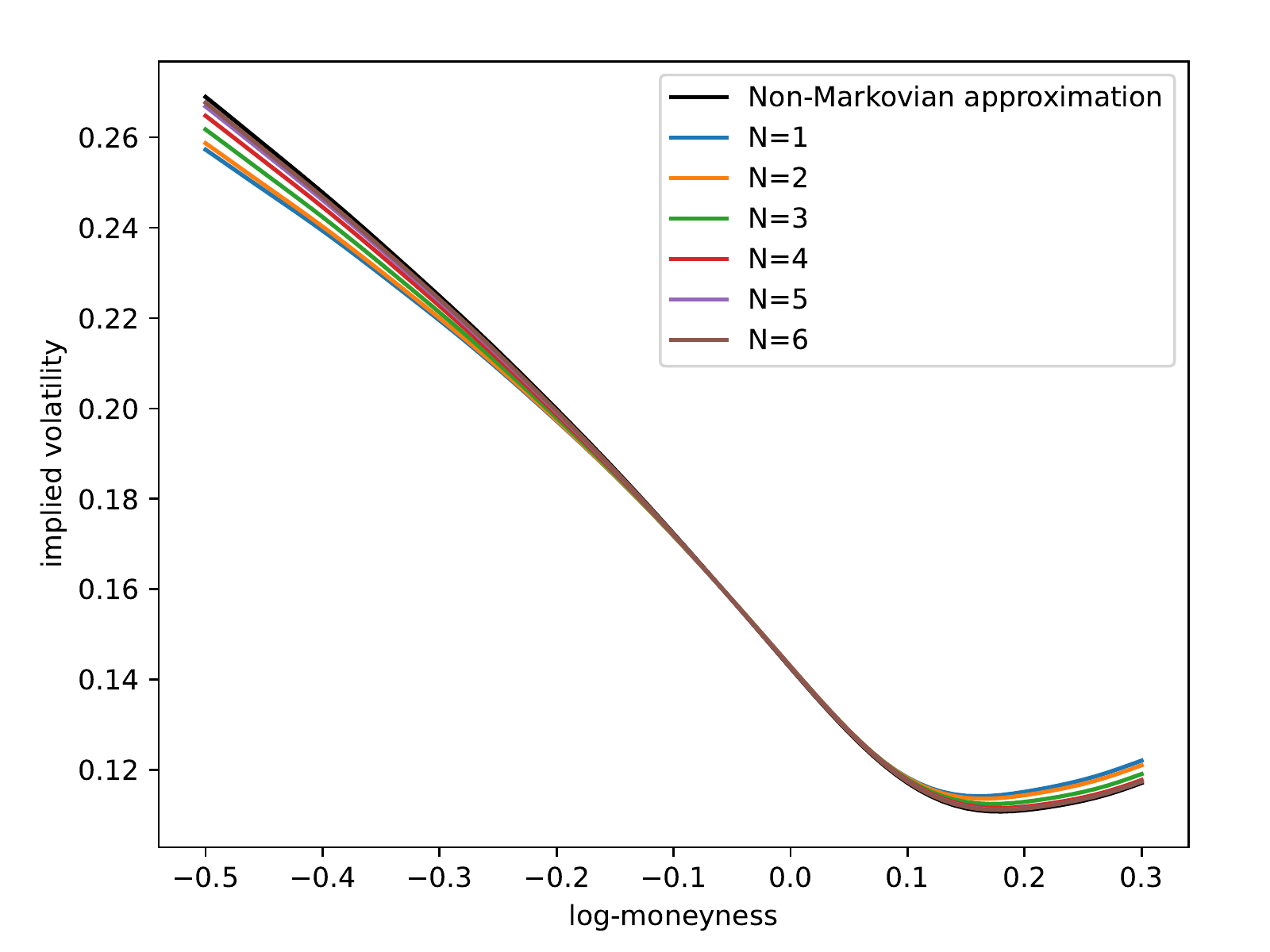}
        \caption{Approximations based on Theorem~\ref{thm:TheGrandAKTheoremWithConstants} for different $N$.}
        \label{fig:RHestonSmiles}
    \end{subfigure}~
    \begin{subfigure}[t]{0.49\textwidth}
        \centering
        \includegraphics[width=\textwidth]{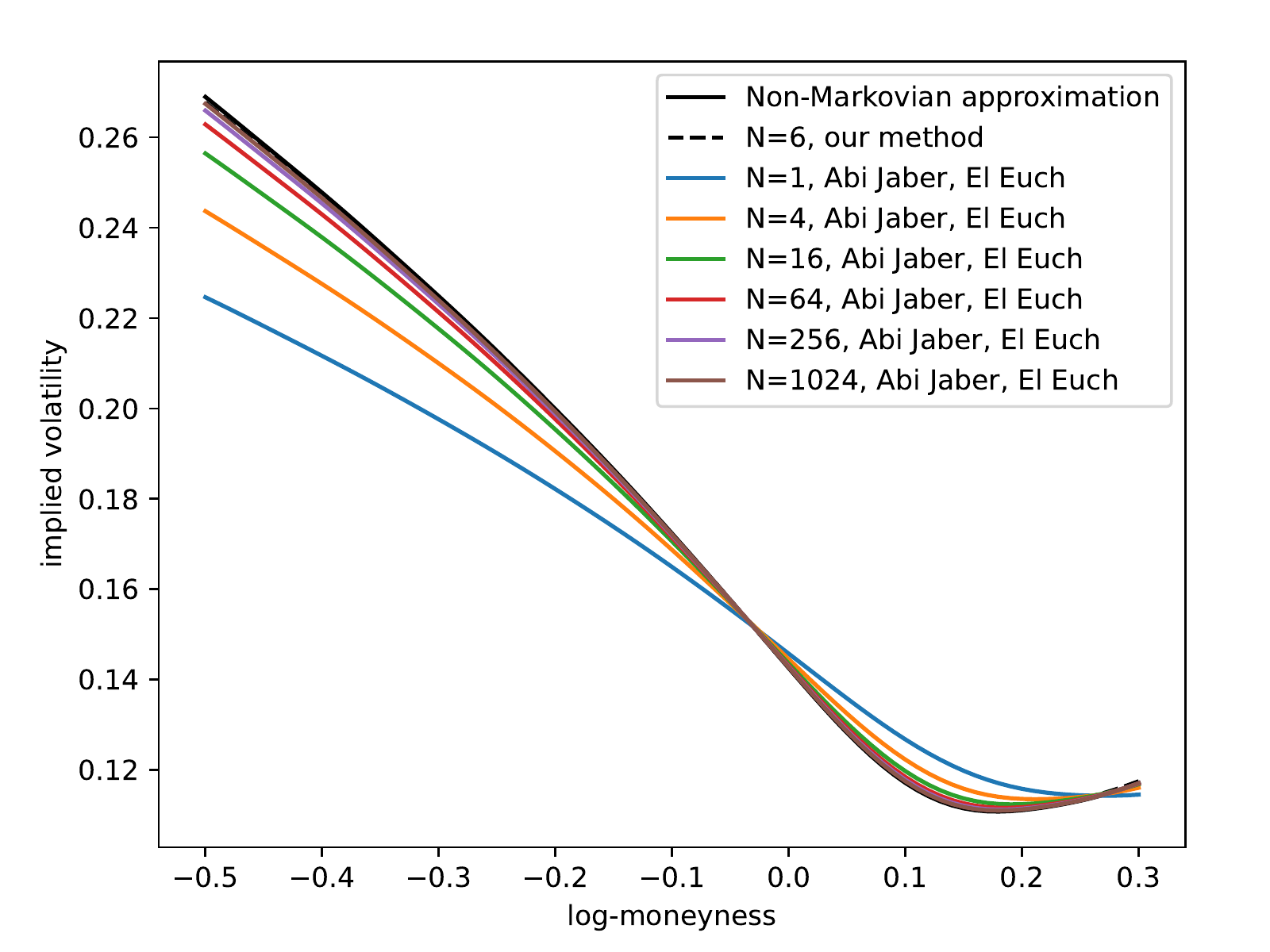}
        \caption{Comparison with other methods from the literature.}
        \label{fig:RHestonSmilesComparison}
    \end{subfigure}
  \caption{Implied volatility smiles for the rough Heston model using different approximations.}
  \label{fig:RBergomi}
\end{figure}

Since the approximations are all already quite accurate, the lines are barely distinguishable. However, it turns out that our approximation with $N=6$ (that is, using 7 nodes if one counts the free node at 0) is better than the approximation in \cite{A2019} with 1024 points. Since we were able to prove a superpolynomial rate of convergence in Corollary \ref{cor:rHestonCallErrorBound}, and the point set in \cite{A2019} has only a convergence rate of $N^{-4H/5}$, we expect our point set to do comparatively even better for larger $N$.

Since the smiles are barely distinguishable, we also plot the errors of our method and the one used in \cite{A2019}. Here, we simply use the error functional $$\textup{err} = \bigg(\int_{-0.5}^{0.3} |\sigma(k, T) - \widehat{\sigma}(k,T)|^2 dk\bigg)^{1/2},$$ where $\sigma$ and $\widehat{\sigma}$ are the implied volatility for the European option with the above mentioned parameters under the rough Heston model and its approximation, respectively. The parameter $k$ refers to the log-moneyness. The errors are shown in Figure \ref{fig:RHestonErrors}.

\begin{figure}
\centering
\includegraphics[width=0.7\columnwidth]{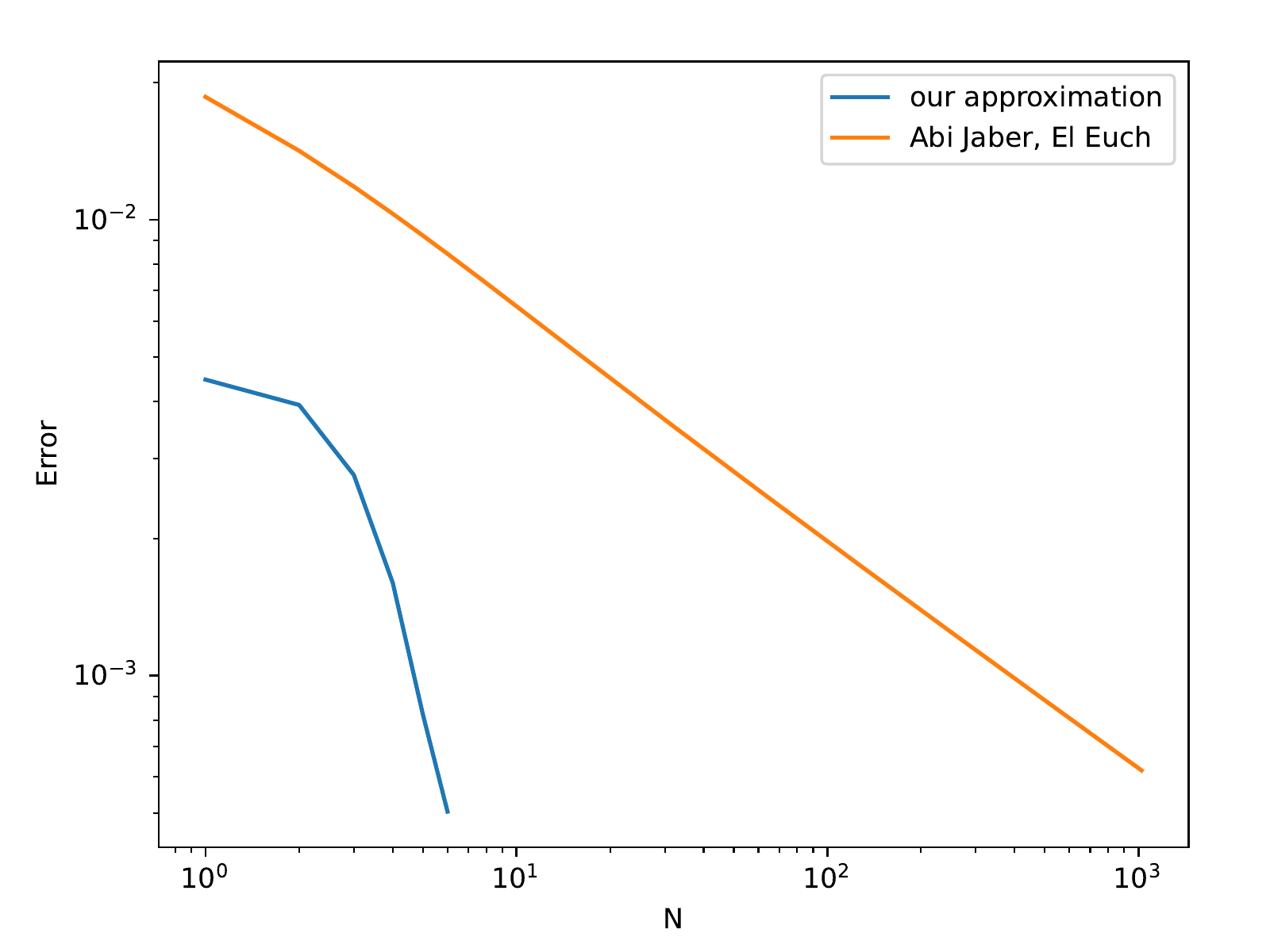}
\caption{Comparison of the approximations errors of the implied volatility smile using our point set and the point set in \cite{A2019}.}
\label{fig:RHestonErrors}
\end{figure}

The time in seconds it took to generate these smiles for different values of $N$ is given in Table \ref{tab:TimeForRHeston}. Here $N=\infty$ refers to the non-Markovian model. This notation is reasonable, as the Markovian approximations converge to the non-Markovian model as $N\to\infty$. The times for our approximation and the approximation by Abi Jaber and El Euch were essentially indistinguishable.

\begin{table}
\centering
\begin{tabular}{c|c|c|c|c|c|c|c|c|c|c|c|c}
N & 1 & 2 & 4 & 8 & 16 & 32 & 64 & 128 & 256 & 512 & 1024 & $\infty$ \\ \hline
Time & 3105 & 3119 & 3102 & 3135 & 3241 & 3393 & 3592 & 3866 & 4387 & 5458 & 7384 & 18567
\end{tabular}
\caption{Time in seconds to generate the implied volatility smiles for different values of $N$. Here, $N=\infty$ refers to the method in \cite{E2019b}.}
\label{tab:TimeForRHeston}
\end{table}

\appendix

\section{Gaussian quadrature}\label{sec:GaussianQuadrature}

In this section, we recall some basic facts about Gaussian quadrature. A more extensive treatment can be found in many books on numerical integration, e.g. in \cite{B2011}.

Let $[a,b]$ be a finite, non-degenerate interval, let $w:[a,b]\to\R$ be a continuous positive weight function, and let $f:[a,b]\to\R$ be a function that we wish to integrate. Then, we consider the approximation 
\begin{equation}\label{eqn:QuadratureRule}
\sum_{i=1}^m w_i f(x_i) \approx \int_a^b f(x) w(x) dx,
\end{equation}
using a quadrature rule with nodes $(x_i)_{i=1}^m$ that lie in $[a,b]$, and weights $(w_i)_{i=1}^m$, where ``$\approx$'' has the usual meaning of ``is approximately equal to''. Note that we have $2m$ degrees of freedom in the choice of our nodes and weights. The Gaussian quadrature rule of level $m$ is the unique choice of nodes and weights that integrates all polynomials of degree at most $2m-1$ exactly. This means that for all polynomials of degree at most $2m-1$, we have equality in \eqref{eqn:QuadratureRule}.

The Gaussian quadrature rule exists for all continuous weight functions $w$. We now give a rough sketch how the nodes and weights can be computed. First, one needs to find orthogonal polynomials $p_n$ of degree $n$. Orthogonality is meant with respect to the scalar product $$\langle f,g\rangle = \int_a^b f(x) g(x) w(x) dx.$$ Such orthogonal polynomials can be found using a recurrence relation, e.g. $$p_{r+1}(x) = \bigg(x-\frac{\langle xp_r, p_r\rangle}{\langle p_r, p_r\rangle}\bigg) p_r(x) - \sum_{j=0}^{r-1} \frac{\langle xp_r, p_j\rangle}{\langle p_j, p_j\rangle}p_j(x),$$ where $p_0(x) = 1$. The nodes for the level $m$ quadrature rule are then the $m$ roots of $p_m$. It can be shown that these roots are all real and lie in $[a,b]$. The weight $w_i$ corresponding to the node $x_i$ can then be computed as $$w_i = \frac{\langle p_{m-1}, p_{m-1}\rangle}{p_m'(x_i)p_{m-1}(x_i)}.$$ An efficient algorithm for the computation of these nodes and weights is, for example, the Golub-Welsch algorithm with a complexity of $\mathcal{O}(m^2)$. 


Of course, good error bounds for numerical integration in general, and Gaussian quadrature in particular, are of central importance. Below we give such an error bound. The idea is essentially, that if the function $f$ we wish to integrate is $(2m)$-times continuously differentiable, then we can approximate $f$ by its Taylor expansion of order $2m$, and the corresponding Taylor polynomial of order $2m-1$ will be integrated exactly by the definition of Gaussian quadrature. In the results below, this is illustrated by the use of the Peano kernel, which, loosely speaking, is the error of the Gaussian quadrature rule in approximating the integral of $x\mapsto x^{2m}.$

\begin{lemma}\label{lem:PeanoRepresentation}\cite[Theorem 4.2.3]{B2011}.
Let $f:[a, b]\to\R$ be a $(2m)$-times continuously differentiable function, and let $\tilde{w}_i$ be the weights and $\tilde{x}_i$ be the nodes of the Gaussian quadrature rule for $i=1, \dots, m$. Then, $$c_H\int_a^b f(x) x^{-H-1/2} dx - \sum_{i=1}^m \tilde{w}_i f(\tilde{x}_i) = \int_a^b f^{(2m)}(x) K_{2m}(x) dx,$$ where $K_{2m}$ is the Peano kernel (corresponding to the weight function $w$). 
\end{lemma}

\begin{lemma}\label{lem:PeanoKernelBound}\cite[Theorem 2]{B1993}.
The above Peano kernel $K_{2m}:[a,b]\to\R$ satisfies 
\begin{align*}
\sup_{x\in[a,b]} |K_{2m}(x)| &\le \frac{(2\pi)^{2m}}{(2m)!}\bigg(\frac{b-a}{2}\bigg)^{2m}\sup_{x\in[-1,1]}|B_{2m}(x)| \sup_{x\in[a,b]} |w(x)|,
\end{align*}
where $B_{2m}$ is the Bernoulli function given by $$B_s(x) = -2\sum_{k=1}^\infty \frac{\cos \big(2\pi kx -\frac{\pi s}{2}\big)}{(2\pi k)^s}.$$
\end{lemma}

Using the fact that the Bernoulli functions agree with the Bernoulli polynomials on $[0,1]$, we have the following lemma.

\begin{lemma}\label{lem:BernoulliBound}\cite[Theorem 1]{L1940}.
For even $s$, we have $$\sup_{x\in[0,1]} |B_s(x)| = \frac{2\zeta(s)}{(2\pi)^s},$$ where $\zeta$ is the Riemann zeta function.
\end{lemma}

\section{Proof of Theorem \ref{thm:TheGrandAKTheoremGeneral}}\label{sec:ProofOfTheoremGeneral}

First, we give the choice of the point set in Theorem \ref{thm:TheGrandAKTheoremGeneral}. It is of the same type as the point set in Theorem \ref{thm:TheGrandAKTheorem}, except that we now choose

\begin{align*}
A &\coloneqq A_{\gamma,\delta} \coloneqq \bigg(\frac{1}{\delta-1/2} + \frac{1}{2-\gamma}\bigg)^{1/2},\quad \xi_0 \coloneqq \exp\bigg(-\frac{\alpha}{(2-\gamma)A}\sqrt{N}\bigg), \quad \xi_n \coloneqq b\exp\bigg(\frac{\alpha}{(\delta-1/2)A}\sqrt{N}\bigg).
\end{align*}

Furthermore, $\alpha$ and $\beta$ are as in Theorem \ref{thm:TheGrandAKTheorem}. Note that this agrees with the choice in Theorem \ref{thm:TheGrandAKTheorem} in the case $\gamma=\delta=H+1/2$. 

The proof of Theorem \ref{thm:TheGrandAKTheoremGeneral} is now essentially analogous to the proof of Theorem \ref{thm:TheGrandAKTheorem}, except that we slightly modify Lemma \ref{lem:IntegrationErrorOnOneInterval} and Lemma \ref{lem:IntegrationErrorOnAllIntervals}. First, Lemma \ref{lem:IntegrationErrorOnOneInterval} is replaced by the following lemma.

\begin{lemma}\label{lem:IntegrationErrorOnOneIntervalGeneral}
Let $\tilde{w}_i$ be the weights and $\tilde{x}_i$ be the nodes of the Gaussian quadrature rule for $i=1,\dots, m$ on the interval $[a,b]$ with respect to the weight function $w$ of type $(\gamma,\delta)$. Then we have the following error bounds.
\begin{enumerate}
\item If $b\le e^{\alpha\beta}$, then $$\bigg|\int_a^b e^{-tx} w(x) dx - \sum_{i=1}^m \tilde{w}_i e^{-t\tilde{x}_i}\bigg| \le Ct^{\gamma-1}\frac{m^{1/2-\gamma}}{2^{2m}} \bigg(\frac{b}{a}-1\bigg)^{2m+1}.$$
\item If $a \ge e^{-\alpha\beta}$, then $$\bigg|\int_a^b e^{-tx} w(x) dx - \sum_{i=1}^m \tilde{w}_i e^{-t\tilde{x}_i}\bigg| \le Ct^{\delta-1}\frac{m^{1/2-\delta}}{2^{2m}} \bigg(\frac{b}{a}-1\bigg)^{2m+1}.$$
\end{enumerate}
Here, $C$ is a constant depending only on $w$.
\end{lemma}

The proof of Lemma \ref{lem:IntegrationErrorOnOneIntervalGeneral} is analogous to the proof of Lemma \ref{lem:IntegrationErrorOnOneInterval}. Additionally, Lemma \ref{lem:IntegrationErrorOnAllIntervals} is replaced by the following lemma.

\begin{lemma}\label{lem:IntegrationErrorOnAllIntervalsGeneral}
In the setting of Theorem \ref{thm:TheGrandAKTheoremGeneral}, we have
$$\int_0^T \bigg|\int_{\xi_0}^{\xi_n} e^{-tx} w(x) dx - \sum_{i=1}^N w_i e^{-tx_i}\bigg|^2 dt \le C\frac{n^2m^{1-2(\gamma\land \delta)}}{2^{4m}} \bigg(e^{\alpha\beta}-1\bigg)^{4m+2}.$$
\end{lemma}

The proof of Lemma \ref{lem:IntegrationErrorOnAllIntervalsGeneral} is again almost the same as the proof of Lemma \ref{lem:IntegrationErrorOnAllIntervals}. The only difference is that one first needs to note that we always have $\xi_{i+1}/\xi_i = e^{\alpha\beta},$ and that thus, for all intervals $[\xi_i,\xi_{i+1}]$, one of the two cases of Lemma \ref{lem:IntegrationErrorOnOneIntervalGeneral} applies. Afterwards, we proceed as in the proof of Lemma \ref{lem:IntegrationErrorOnAllIntervals}.

Finally, the proof of Theorem \ref{thm:TheGrandAKTheoremGeneral} follows the same lines as the proof of Theorem \ref{thm:TheGrandAKTheorem}. We also remark that an optimization procedure similar to the one in Theorem \ref{thm:TheGrandAKTheoremWithConstants} can be performed to improve the polynomial rate in $N$ in Theorem \ref{thm:TheGrandAKTheoremGeneral}.

\section{Proof of Theorem \ref{thm:TheGrandAKTheoremWithConstants}}\label{sec:ProofOfTheoremWithConstants}

This proof is split in three steps. First, we determine the dependence of $a$, $b$ on $T$. In the second step, we optimize $a$ and $b$ over $H$ and $N$. In this step, we make some additional assumptions, which we verify in the last step.

\noindent
\textbf{Step 1:} Let us first consider the dependence of $a$ and $b$ on $T$. Considering the proof of Theorem \ref{thm:TheGrandAKTheorem}, it is not so difficult to see that we have 
\begin{align*}
\E\big|X_T - \hat{X}_T\big|^2 &\le C c_H^2 \Bigg(\bigg(\frac{T^3}{(3/2-H)^2}a^{3-2H} + \frac{3}{2H^2}b^{-2H}\bigg)\exp\Big(-\frac{2\alpha}{A}\sqrt{N}\Big)\\
&\qquad + \frac{5\pi^3}{12}\frac{A^{2-2H}T^{2H}}{\beta^{2-2H}H}N^{1-H}\bigg(\frac{1}{2}\bigg(\Big(\frac{b}{a}\Big)^{1/n}e^{\alpha\beta}-1\bigg)\bigg)^{4m+2}\Bigg).
\end{align*}
This bound can be made homogeneous in $T$ by choosing $a = a_1T^{-1}$ and $b=b_1T^{-1}$. With some abuse of notation, we again call $a_1$ and $b_1$ $a$ and $b$, respectively.

\textbf{Step 2:} Now, we want to determine the dependence of $a$ and $b$ on other variables. Denote $c \coloneqq \Big(\frac{b}{a}\Big)^{1/n}.$ Then,
\begin{align*}
\bigg(\frac{1}{2}\bigg(ce^{\alpha\beta}-1\bigg)\bigg)^{4m+2} &= \bigg(\frac{1}{2}\bigg(\frac{ce^{\alpha\beta}-1}{e^{\alpha\beta} - 1}\Big(e^{\alpha\beta}-1\Big)\bigg)\bigg)^{4m+2}\\
&= \bigg(\frac{1}{2}\Big(e^{\alpha\beta}-1\Big)\bigg)^{4m+2}\bigg(\frac{ce^{\alpha\beta}-1}{e^{\alpha\beta} - 1}\bigg)^{4m+2}.
\end{align*}

Therefore, by our choice of $\alpha$ and $\beta$,
\begin{align*}
\E\big|X_T - \hat{X}_T\big|^2 &\le C c_H^2T^{2H} \Bigg(\frac{1}{(3/2-H)^2}a^{3-2H} + \frac{3}{2H^2}b^{-2H}\\
&\qquad + \frac{5\pi^3}{48}\Big(e^{\alpha\beta}-1\Big)^2\frac{A^{2-2H}}{\beta^{2-2H}H}N^{1-H}\bigg(\frac{ce^{\alpha\beta}-1}{e^{\alpha\beta} - 1}\bigg)^{4m+2}\Bigg)\exp\Big(-\frac{2\alpha}{A}\sqrt{N}\Big).
\end{align*}

Note that $$\bigg(\frac{ce^{\alpha\beta}-1}{e^{\alpha\beta} - 1}\bigg)^{4m+2} = \bigg(1 + (c-1)\frac{e^{\alpha\beta}}{e^{\alpha\beta} - 1}\bigg)^{4m+2}.$$ A simple argument shows that it is advisable to choose $c<1$, i.e. $b<a$. Then, as $c<1$, we have $\log c < 0$, and, hence, $$c = e^{\log c} \le 1 + \log c + \frac{(\log c)^2}{2}.$$ Then,
\begin{align*}
\bigg(1 + (c-1)\frac{e^{\alpha\beta}}{e^{\alpha\beta} - 1}\bigg)^{4m+2} &\le \bigg(1 +\frac{e^{\alpha\beta}}{e^{\alpha\beta} - 1}\bigg(\log c + \frac{(\log c)^2}{2}\bigg)\bigg)^{4m+2}\\
&= \bigg(1 + \frac{e^{\alpha\beta}}{e^{\alpha\beta} - 1}\bigg(\frac{\beta}{A\sqrt{N}}\log\frac{b}{a} + \frac{\beta^2}{2A^2N}\bigg(\log\frac{b}{a}\bigg)^2\bigg)\bigg)^{4m+2}.
\end{align*}

Assume that $N$ is so large that $$-\frac{\beta}{A\sqrt{N}}\log\frac{b}{a} \ge \frac{\beta^2}{A^2N}\bigg(\log\frac{b}{a}\bigg)^2,$$ i.e.
\begin{equation}\label{eqn:NAssumption}
N \ge \bigg(\frac{\beta}{A}\log\frac{b}{a}\bigg)^2.
\end{equation}
Then, 
\begin{align*}
\bigg(1 + \frac{e^{\alpha\beta}}{e^{\alpha\beta} - 1}\bigg(\frac{\beta}{A\sqrt{N}} &\log\frac{b}{a} + \frac{\beta^2}{2A^2N}\bigg(\log\frac{b}{a}\bigg)^2\bigg)\bigg)^{4m+2}\\
&\le \bigg(1 + \frac{e^{\alpha\beta}}{8(e^{\alpha\beta} - 1)}\frac{4\beta}{A\sqrt{N}}\log\frac{b}{a}\bigg)^{\frac{4\beta}{A}\sqrt{N}+2}\\
&\le \exp\bigg(\frac{e^{\alpha\beta}}{8(e^{\alpha\beta} - 1)}\log\frac{b}{a}\bigg)\bigg(1 + \frac{e^{\alpha\beta}}{8(e^{\alpha\beta} - 1)}\frac{4\beta}{A\sqrt{N}}\log\frac{b}{a}\bigg)^2\\
&\le \bigg(\frac{b}{a}\bigg)^{\frac{e^{\alpha\beta}}{8(e^{\alpha\beta} - 1)}}.
\end{align*}

Thus,
\begin{align*}
\E\big|X_T - \hat{X}_T\big|^2 &\le C c_H^2T^{2H} \Bigg(\frac{1}{(3/2-H)^2}a^{3-2H} + \frac{3}{2H^2}b^{-2H}\\
&\qquad + \frac{5\pi^3}{48}\Big(e^{\alpha\beta}-1\Big)^2\frac{A^{2-2H}}{\beta^{2-2H}H}N^{1-H}\bigg(\frac{b}{a}\bigg)^{\frac{e^{\alpha\beta}}{8(e^{\alpha\beta} - 1)}}\Bigg)\exp\Big(-\frac{2\alpha}{A}\sqrt{N}\Big).
\end{align*}

We now want to solve the optimization problem $$\min_{0<b\le a} A_1 b^{-q_1} + A_2 \bigg(\frac{b}{a}\bigg)^{q_2} + A_3 a^{q_3}.$$

This is minimized in
\begin{align*}
a &= \bigg((q_1A_1)^{q_2}(q_2A_2)^{q_1}(q_3A_3)^{-(q_1+q_2)}\bigg)^{\frac{1}{q_1q_2+q_1q_3+q_2q_3}},\\
b &= \bigg((q_1A_1)^{q_2+q_3}(q_2A_2)^{-q_3}(q_3A_3)^{-q_2}\bigg)^{\frac{1}{q_1q_2+q_1q_3+q_2q_3}},
\end{align*}


where,
\begin{align*}
&q_1 = 2H,\qquad\qquad q_2 = \frac{e^{\alpha\beta}}{8(e^{\alpha\beta}-1)},\qquad\qquad\qquad\qquad\qquad\qquad\quad\ q_3 = 3-2H,\\
&A_1 = \frac{3}{2H^2},\qquad\quad A_2 = \frac{5\pi^3}{48}\Big(e^{\alpha\beta}-1\Big)^2\frac{A^{2-2H}}{\beta^{2-2H}H}N^{1-H},\qquad\quad A_3 = \frac{1}{(3/2-H)^2}.
\end{align*}

Plugging in these values, we get precisely the statement of the theorem.

\textbf{Step 3:} We now need to verify that $b<a$, and that \eqref{eqn:NAssumption} holds. We start with $b<a$. Indeed,
\begin{align*}
\frac{b}{a} &= \bigg(\Big(\frac{5\pi^3}{1152}e^{\alpha\beta}\Big(e^{\alpha\beta}-1\Big)\frac{A^{2-2H}}{\beta^{2-2H}}N^{1-H}\Big)^{-3}\Big(\frac{H}{3(3/2-H)}\Big)^{2H}\bigg)^{\frac{1}{\frac{3e^{\alpha\beta}}{8(e^{\alpha\beta}-1)} + 6H - 4H^2}} \le 1
\end{align*}
if and only if
$$N \ge \Big(\frac{5\pi^3}{1152}e^{\alpha\beta}\Big(e^{\alpha\beta}-1\Big)\frac{A^{2-2H}}{\beta^{2-2H}}\Big)^{-3}\Big(\frac{H}{3(3/2-H)}\Big)^{\frac{2H}{3-3H}}.$$
The right hand side can be uniformly (over $H$) bounded by $2$. Hence, it suffices to choose $N\ge 2$.

Furthermore, we have made the assumption $$N\ge \bigg(\frac{\beta}{A}\log\frac{b}{a}\bigg)^2$$ in \eqref{eqn:NAssumption}, which we also need to verify. We have
\begin{align*}
\bigg|\frac{\beta}{A}\log\frac{b}{a}\bigg| &= \bigg|\frac{\beta}{A}\frac{1}{\frac{3e^{\alpha\beta}}{8(e^{\alpha\beta}-1)} + 6H - 4H^2}\log\bigg(\Big(\frac{5\pi^3}{1152}e^{\alpha\beta}\Big(e^{\alpha\beta}-1\Big)\frac{A^{2-2H}}{\beta^{2-2H}}N^{1-H}\Big)^{-3}\Big(\frac{H}{3(3/2-H)}\Big)^{2H}\bigg)\bigg|\\
&\le \frac{\beta}{A}\frac{3}{\frac{3e^{\alpha\beta}}{8(e^{\alpha\beta}-1)} + 6H - 4H^2}\log\bigg(\frac{5\pi^3}{1152}e^{\alpha\beta}\Big(e^{\alpha\beta}-1\Big)\frac{A^{2-2H}}{\beta^{2-2H}}N\bigg)\\
&\le \frac{8\beta(e^{\alpha\beta}-1)}{Ae^{\alpha\beta}}\log\bigg(\frac{5\pi^3}{1152}e^{\alpha\beta}\Big(e^{\alpha\beta}-1\Big)\frac{A^2}{\beta^2}\bigg) + \frac{8\beta(e^{\alpha\beta}-1)}{Ae^{\alpha\beta}}\log N.
\end{align*}
Plugging in our values of $\alpha$ and $\beta$, we get
\begin{align*}
\bigg(\frac{\beta}{A}\log\frac{b}{a}\bigg)^2 &\le \bigg(\frac{1.26}{A}\log\big(0.67A^2\big) + \frac{1.26}{A}\log N\bigg)^2 \le \big(0.8 + 0.73\log N\big)^2.
\end{align*}
Now, $$N \ge \big(0.8 + 0.73\log N\big)^2$$ holds for all $N\ge 1$. This shows the theorem.

\section{Learning a better rate of convergence}\label{sec:AppendixLearningBetterRate}

We describe here how we got the results in Section \ref{sec:LearningBetterRate}. We make the ansatz 
\begin{align}
\xi_0 &= C_1T^{-1}\exp\bigg(-\frac{\alpha}{(3/2-H)A}\sqrt{N}\bigg), \label{eqn:BestXi0}\\
\xi_n &= C_2T^{-1}\exp\bigg(-\frac{\alpha}{HA}\sqrt{N}\bigg), \label{eqn:BestXin}\\
m &= \frac{\beta}{A}\sqrt{N}, \label{eqn:Bestm}\\
\int_0^T \big|G(t) - \hat{G}(t)\big|^2 dt &= C_3T^{2H}\exp\bigg(-\frac{2\alpha}{A}\sqrt{N}\bigg), \label{eqn:BestError}
\end{align}
and our goal is to learn $\alpha$ and $\beta$ in particular, and get some approximation for $C_1$, $C_2$ and $C_3$, which may depend on $H$ and $N$. Here, \eqref{eqn:BestError} philosophically has a different meaning, as the left-hand side, i.e. the true error, is determined by our choice for $\xi_0$, $\xi_n$ and $m$, and need not be of the form of the right-hand side of \eqref{eqn:BestError}. Nonetheless, we still choose to use the errors as well to estimate $\alpha$ with equation \eqref{eqn:BestError}.

To do so, we proceed as follows. We take the error representation in Proposition \ref{prop:fBmErrorRepresentation} for $T=1$. Then, we pick $$H\in\{0.05, 0.1, 0.15, 0.2, 0.25, 0.3, 0.35, 0.4, 0.45\}$$ and $N$ in a range from $1$ to $1024$. For each of these choices of $H$ and $N$, we compute the optimal $m$, $\xi_0$ and $\xi_n$, as well as the corresponding error, using an optimization algorithm. Then we use this data to estimate $\alpha$ and $\beta$.

First, using the optimal values of $m$, we can determine the corresponding values of $\beta$ by reformulating equation \eqref{eqn:Bestm}. Empirically, we observed that paths of the estimated values of $\beta$ for different choices of $H$ and $N$ indeed converge to the same number, independent of $H$, and we estimate $\beta = 0.9$.

For $\alpha$, we can use $\xi_0,$ $\xi_n$ and the error as estimators. In other words, we set $C_i = 1$, and solve equations \eqref{eqn:BestXi0}, \eqref{eqn:BestXin}, and \eqref{eqn:BestError} for $\alpha$. Again we empirically observe that all the paths seem to converge to some uniform number, which we estimate as $\alpha=1.8$.

We can now plug these values for $\alpha$ and $\beta$ into equations \eqref{eqn:BestXi0}-\eqref{eqn:BestError} and apply regression to get estimates for the constants $C_1, C_2$ and $C_3$. We get the results of Section \ref{sec:LearningBetterRate}.

\printbibliography


\end{document}